\documentclass[12pt,a4paper]{article}
\usepackage[left=2cm,right=2cm,top=2.5cm,bottom=2.5cm]{geometry}
\usepackage{enumitem}
\usepackage{lastpage}
\usepackage{eurosym}
\usepackage{hyperref}
\usepackage{fancyhdr}

\usepackage{algorithmicx}
\usepackage{algorithm}
\usepackage{algpseudocode}

\usepackage{color}

\usepackage[utf8]{inputenc}
\usepackage[english]{babel}
\usepackage{amsmath,mdframed}
\usepackage{mathtools}
\usepackage{amsfonts}
\usepackage{amssymb}
\usepackage{multirow}
\usepackage[dvipsnames]{xcolor}
\usepackage{tabularx}
\usepackage{setspace}

\usepackage{bm}
\usepackage{float}
\usepackage{graphicx}
\usepackage{caption}
\usepackage{subcaption}
\usepackage{comment}
\usepackage{array}
\usepackage{booktabs}
\newcolumntype{x}[1]{>{\centering}p{#1}}
\usepackage{natbib}
\usepackage{rotating}
\usepackage{authblk}

\usepackage{amsthm}

\let\ltxcup\cup
\usepackage{mathabx}

\newtheorem{theorem}{Theorem}

\newtheorem{assumption}{Assumption}

\title{Valid post-selection inference for penalized G-estimation}
\author[1,*]{Ajmery Jaman}
\author[2]{Ashkan Ertefaie}
\author[3,5]{Michèle Bally}
\author[4]{Renée Lévesque}
\author[1]{Robert W.\ Platt}
\author[5,1]{Mireille E.\ Schnitzer}

\affil[1]{{\small Department of Epidemiology, Biostatistics and Occupational Health, McGill University, Montreal, H3A 1G1, Canada}}
\affil[2]{{\small Department of Biostatistics, Epidemiology and Informatics, University of Pennsylvania, Philadelphia, PA19104, U.S.A.}}
\affil[3]{{\small Department of Pharmacy, Centre Hospital of University of Montreal, Montreal, H2X 0C1, Canada}}
\affil[4]{{\small Department of Medicine, University of Montreal, Montreal, H3T 1J4, Canada}}
\affil[5]{{\small Faculty of Pharmacy, University of Montreal, Montreal, H3C 3J7, Canada}}
\affil[*]{{\small Corresponding author: Ajmery Jaman, Email: ajmery.jaman@mail.mcgill.ca}}
\date{}

\begin{document}
\renewcommand{\vec}[1]{\boldsymbol{\mathbf{#1}}}
\maketitle
\begin{abstract}
\noindent
Understanding treatment effect heterogeneity is important for decision-making in medical and clinical practices, or handling various engineering and marketing challenges. When dealing with high-dimensional covariates or when the effect modifiers are not predefined and need to be discovered, data-adaptive selection approaches become essential. However, with data-driven model selection, the quantification of statistical uncertainty is complicated by post-selection inference due to difficulties in approximating the sampling distribution of the target estimator. Data-driven model selection tends to favor models with strong effect modifiers with an associated cost of inflated type I errors. %However, data-driven model selection complicates the quantification of statistical uncertainty in post-selection inference and makes it difficult to approximate the sampling distribution of the target estimator. Such model selection tends to favor models with strong effect modifiers with an associated cost of inflated type I errors.
Although several frameworks and methods for valid statistical inference have been proposed for ordinary least squares regression following data-driven model selection, fewer options exist for valid inference for effect modifier discovery in causal modeling contexts. In this article, we extend two different methods to develop valid inference for penalized G-estimation that investigates effect modification of proximal treatment effects within the structural nested mean model framework. We show the asymptotic validity of the proposed inferential methods. In our simulation study, the proposed methods effectively controlled the false coverage rates for the target parameters, while the naive inference based on the sandwich variance estimator resulted in false coverage rates higher than the nominal level. Our work is motivated by the study of hemodiafiltration for treating patients with end-stage renal disease at the Centre Hospitalier de l'Université de Montréal. We apply these methods to draw inference about the effect heterogeneity of dialysis facility on the repeated session-specific hemodiafiltration outcomes.\\

\noindent
%in alphabetical order
\end{abstract}

\setstretch{1}
\section{Introduction}

Understanding treatment effect heterogeneity is important for decision-making in medical and clinical practices, or addressing various engineering and marketing challenges. When dealing with high-dimensional covariates or when effect modifiers are not predefined and need to be discovered, data-adaptive selection approaches become essential. However, with data-driven model selection, the quantification of statistical uncertainty is complicated by post-selection inference. Classical inference is built on a framework where all modeling decisions are made independently of the data from which inference is drawn. The asymptotic distribution of the parameter estimator is challenging to derive due to the non-negligible estimation bias and sparsity effects associated with the high dimensional parameters \citep{ning2017general}. If we perform data-driven effect modifier selection, we tend to favor models with strong effect modifiers with an associated cost of inflated type I errors \citep{zhao2022selective}. %This refers to the multiplicity problem in regression that gets more severe with increasing number of measured covariates. 
Data-driven selection procedures produce a model that is itself stochastic, and this model selection uncertainty is not accounted for by the classical inference theory.

In the last fifteen years, there have been several proposed frameworks and methods for valid statistical inference following data-driven model selection. \cite{tibshirani2016exact} and \cite{lee2016exact} developed frameworks for inference under forward stepwise regression, least angle regression, and the Least Absolute Shrinkage and Selection Operator (LASSO). These conditional approaches provide valid inference only for a few specific model selection methods, not for generic variable selection, and are dependent on distributional assumptions for the response. The Post-Selection Inference (PoSI) method, proposed by \cite{berk2013valid} and later generalized by \cite{bachoc2020uniformly}, % which relax the ``Gaussian homoscedastic model" assumption, 
provides inferential guarantees for arbitrary model selection approaches, including informal ones. The PoSI method ensures valid inference even when an incorrect model is selected, but this inferential procedure is computationally expensive \citep{kuchibhotla2020valid}. By addressing the limitations of the PoSI method and accommodating misspecification of the normal linear model, \cite{kuchibhotla2020valid} introduced the Universal Post Selection Inference (UPoSI) approach for OLS regression assuming either fixed or random covariates, which are referred to as fixed-design UPoSI and random-design UPoSI, respectively.  %, and discussed the method for independent as well as some types of dependent data. The computational cost of the UPoSI approach is proportional to the number of covariates, which makes it an attractive method for valid post selection inference. Moreover, this approach does not impose any assumption on the correct specification of the model being used, thus making it a ``model-robust" inferential procedure. Such inference procedures are valid regardless of the variable selection method such as formal, informal, \textit{post hoc}, fully or only partly specified.
The UPoSI approach is computationally efficient--its cost is proportional to the number of covariates--and does not require correct model specification, making it a ``model-robust" inferential procedure. Other advancements include the debiased or desparsifying method proposed by \cite{zhang2014confidence}, known as the low dimensional projection estimator (LDPE), which constructs confidence intervals for linear or generalized linear models with the Lasso penalty. \cite{ning2017general} proposed a decorrelated score test for inference in penalized M-estimation. Unlike the work in \cite{zhang2014confidence} which are tailored for individual models,  the decorrelated score method \citep{ning2017general} provides a general framework for high dimensional inference that can be used to infer the oracle parameter under misspecified models. Based on the idea of projected estimating equations \citep{zhang2014confidence}, \cite{xia2022statistical} developed an inference procedure for linear functionals of high-dimensional longitudinal data using generalized estimating equations (GEE). However, there are few contributions on valid inference for effect modifier discovery in causal modeling contexts. In one such work, \cite{zhao2022selective} adapted the approach of \cite{lee2016exact} and proposed a conditional selective inference procedure for effect modification estimated using the LASSO. In the context of multistage decision problems of dynamic treatment regimes, \cite{jones2022valid} extended the UPoSI approach \citep{kuchibhotla2020valid} to develop valid inference for robust Q-learning.  \cite{gao2025asymptotic} adapted the approach of \cite{ning2017general} to develop an asymptotic inference method for multistage stationary treatment policies in the presence of high dimensional feature variables.

In this paper, we contribute to this growing body of work by developing valid inferential methods for the recently proposed doubly-robust penalized G-estimation \citep{jaman2025penalized}. This method estimates the proximal effects of exposure with simultaneous data-adaptive selection of effect modifiers within a structural nested mean model (SNMM) framework, particularly for repeated outcomes. We present two different proposals for valid inference on effect modification of proximal treatment effects: one is an extension of the UPoSI approach \citep{kuchibhotla2020valid} and another is based on the decorrelated score test introduced by \cite{ning2017general}. %ESRD is the final, permanent stage of chronic kidney disease. In patients with ESRD, either kidney transplant or regular dialysis are necessary for survival. 
Our methodological development is motivated by a study of hemodiafiltration (HDF), a dialysis technique for treating patients with end-stage renal disease (ESRD). %HDF combines two processes \cite{ronco2007hemodiafiltration}: diffusion (where solute molecules passively move from the blood to the dialysis fluid) and convection (where larger molecules are cleared from the blood).
Hemodiafiltration is the standard treatment for ESRD patients at the University of Montreal Hospital Centre (CHUM) outpatient dialysis clinic and its affiliated ambulatory dialysis center (CED).  Hemodiafiltration cleans waste and excess fluids from the blood by combining diffusive clearance and convective removal of solutes \citep{ronco2007hemodiafiltration}. It involves the ultrafiltration of a large volume of plasma water, which requires substitution fluid to be administered to the patient to preserve fluid balance. %Convection volume is calculated as the sum of the substitution volume and the ultrafiltration volume \cite{marcelli2015high}.
Dialysis effectiveness is indicated by the convection volume attained during each session, which is calculated as the sum of the substitution volume and the ultrafiltration volume \citep{marcelli2015high}. 
%Results of the CONvective TRAnsport STudy \cite{grooteman2012effect} randomized controlled trial and a meta-analysis of individual patient-level data from randomized controlled trials suggest that the target convection volume should be at least 24 liters per session \cite{chapdelaine2015optimization}. Observations of consistently lower average convection volumes at the outpatient dialysis clinic located at the CHUM compared to those at the CED triggered an interest to investigate a potential effect of the dialysis facility (CHUM vs.\ CED) on the convection volume or to identify the group of patients for whom such an effect exists. Hospital records provide data that include sociodemographic information, diagnoses, medications, blood test results, and dialysis treatment parameters for all dialysis sessions of each patient between March 1st, 2017 and December 1st, 2021. Availability of the convection volume outcome of each successful dialysis session allows us to explore the impact of dialysis facility on the session-specific mean convection volume.
Using the data extracted from hospital records, \cite{jaman2025penalized} explored the effect heterogeneity of the dialysis facility (CHUM vs.\ CED) on the session-specific mean convection volumes applying penalized G-estimation. In this paper, we apply our proposed methods to provide valid post-selection inference for the effects estimated by \cite{jaman2025penalized}. The identification of effect modifiers in \cite{jaman2025penalized} is informed by the data itself. Without proper adjustment for post-selection, the uncertainty in the estimated effect modification could be underestimated, leading to overconfident or misleading conclusions. Hence, addressing post-selection inference in our real-data analysis is essential.

The rest of the paper is organized as follows. In Section~\ref{sec.methodology}, we briefly describe the penalized G-estimator and the two proposed inferential procedures. In Section~\ref{sec.simulation}, we evaluate and compare the finite sample performance of the inferential methods along with naive inference based on a sandwich estimator via a simulation study under misspecification of the  treatment-free model. We then apply these methods to infer the heterogenous effect of dialysis facility on the session-specific hemodiafiltration outcomes in Section~\ref{sec.application}. Finally, we present a discussion in Section~\ref{sec.discussion}.

\section{Methodology} \label{sec.methodology}

\subsection{Notation}

Following \cite{jaman2025penalized}, suppose that we have data from $J$ sequential hemodiafiltration sessions for $n$ different ESRD patients. At each session, we record the outcome, the treatment received, and pre-session covariates. We denote the observed continuous outcome for patient $i$ at session $j$ by $Y_{ij}$, the (binary) treatment received by $A_{ij}$, and the vector of covariates by $\vec L_{ij}$, for all $i=1,\dots,n, j=1,\dots,J$. Let $\vec H_{ij}$ % be a vector that contains a one (corresponding to the main effect of the exposure) and each potential EM of interest (or their functions) chosen from the 
represent the history at occassion $j$ that comprises covariate history $\Bar{\vec L}_{ij}=\{\vec L_{i1},\ldots \vec L_{ij}\}$, past exposures $\Bar{A}_{i(j-1)}=\{A_{i1},\ldots,A_{i(j-1)}\}$ and past outcomes $\Bar{Y}_{i(j-1)}=\{Y_{i1},\ldots,Y_{i(j-1)}\}$. Throughout we use the potential outcomes framework \citep{robins1989analysis}. We define $Y_{ij}(\Bar{a}_{j})$ as the potential outcome that would have been observed at occasion $j$ for patient $i$ if the treatment history $\Bar{A}_{ij}=\{A_{i1},\ldots,A_{ij}\}$ were set counterfactually to $\Bar{a}_j=\{a_1,\ldots,a_j\}$.

\subsection{Proximal effects of treatment}\label{sec.SNMM}

The proximal (short-term) effects of the exposure at measurement occassion $j$ can be modelled using a linear structural nested mean model (SNMM) as follows
\citep{robins1989analysis,vansteelandt2014structural}:
\begin{align} \label{snmm.proximal}
E\{Y_{ij}(\Bar{a}_{j-1}, a_j)-Y_{ij}(\Bar{a}_{j-1}, 0)|\vec H_{ij}=\vec h_{ij},\Bar{A}_{ij}=\Bar{a}_{j}\} = \gamma_j^*(a_{j},\vec h_{ij};\vec\psi), 
\end{align}
where $j=1,\ldots,J$, $\gamma_j^*(a_{j},\vec h_{ij};\vec\psi)$, referred to as the ``treatment blip", is a scalar-valued function defined in terms of parameter $\vec\psi$, $\vec h_{ij}$ represent the realized values for $\vec H_{ij}$, and $\vec\psi=(\psi_0,\psi_1,\ldots,\psi_{K-1})'$ is a K-dimensional vector of parameters. The difference described in (\ref{snmm.proximal}) 
shows the effect of treatment $a_j$ compared to the reference treatment 0 on the outcome at occasion $j$, given the history up to that point. The goal is to estimate the parameters $\vec\psi$ utilizing the observed data via G-estimation \citep{robins2008estimation,vansteelandt2014structural}. 
The core idea of this approach is to construct the $j$-th proximal blipped down outcome, $U_{ij}=Y_{ij} - \gamma_j^*({A}_{ij}, \vec H_{ij};\vec\psi)$, which is a transformation of the observed data such that
%\begin{align*}
%    E[U_{ij}|\vec H_{ij}=\vec h_{ij},A_{ij}=a_{j}]= E[Y_{ij}(\Bar{a}_{j-1}, 0)|\vec H_{ij}=\vec h_{ij},A_{ij}=a_{j}],
%\end{align*}
%i.e., 
it has the same mean as $Y_{ij}(\Bar{a}_{j-1}, 0)$, i.e., the potential outcome under the reference treatment level 0 at occasion $j$. Under the restriction that the blip parameters are the same across measurement occasions, we can parameterize the blip as a simple function of the history as follows \citep{vansteelandt2014structural, boruvka2018assessing}:
$\gamma_j^*(a_{j},\vec h_{ij};\vec\psi)=a_{j}\vec h_{ij}'\vec\psi$, where $\vec h_{ij}$ contains a one and potential confounders (or functions of these) chosen from the histories. Each component of $\vec\psi$ reflects the change in the treatment effect attributable to the corresponding covariate. To ensure consistent parameter estimation under this parametric approach, the blip model must be accurately specified as a function of the history.

\subsection{Effect modifier discovery via penalized estimating equations} \label{sec.penalizedG}
For the estimation of the SNMM parameters with simultaneous selection of effect modifiers, \cite{jaman2025penalized} proposed the penalized G-estimator by adding a nonconvex smoothly clipped absolute deviation (SCAD) penalty \citep{fan2001variable} in the efficient score function \citep{chakraborty2013statistical} of $\vec\psi$. Under the usual causal assumptions \ref{asmp1}-\ref{asmp3} (consistency, sequential ignorability and positivity mentioned in Appendix~\ref{A.assump}) for identifiability of the target parameter $\vec\psi$, \cite{jaman2025penalized} proposed the following penalized efficient score function
\begin{align} \label{pen.score}
    \vec S^{P}(\vec\psi)&= \sum_{i=1}^n \Big\{\frac{\partial\vec\gamma^*(\vec A_i,\vec H_i;\vec\psi)}{\partial\vec\psi'} - E\Big(\frac{\partial\vec\gamma^*(\vec A_i,\vec H_i;\vec\psi)}{\partial\vec\psi'}|\vec H_i\Big)\Big\}'\,Var(\vec U_i|\vec H_i)^{-1} \nonumber \\
    &{\hspace{1.3in}}\{\vec U_i - E(\vec U_i|\vec H_i)\}- n{\vec q}_{\lambda_n}(|\vec\psi|)\text{sign}(\vec\psi),
\end{align}
where $\vec A_i=(A_{i1},\ldots, A_{iJ})^\top$, $\vec H_i=(\vec H_{i1},\ldots,\vec H_{iJ})^\top$ is a $J\times K$ matrix representing the unit-wise history for the $i$-th subject, $\vec U_i = (U_{i1},\ldots,U_{iJ})^\top$, $\text{E}(\vec U_i|\vec H_i)= \vec H_i\vec\delta$ is the treatment-free model with $\vec\delta$ denoting its parameters, $\vec q_{\lambda_n}(|\vec\psi|)=(0, q_{\lambda_n}(|\psi_1|),\ldots,q_{\lambda_n}(|\psi_{K-1}|))'$, $q(.)$ indicates the first-derivative of the SCAD penalty, and $\lambda_n$ is the tuning parameter. \cite{jaman2025penalized} considered a working structure for $\text{Var}(\vec U_i|\vec H_i) = \vec Q_i^{1/2}\,\vec R_i(\rho)\,\vec Q_i^{1/2}$, where $\vec Q_i=\sigma^2\vec I_{(J)}$ and $\vec R_i(\rho)$ is the $J\times J$ matrix representing the correlations among the blipped down outcomes of a patient and is defined with respect to parameter $\rho$. %an idea similar to the generalized estimating equations approach\cite{liang1986longitudinal}.
Some technical aspects related to $\vec R_i(\rho)$ are briefly outlined in Appendix~\ref{A.sec.est.R} and further details regarding the estimation of this correlation matrix can be found in the works of \cite{jaman2016determinant} and \cite{sultana2023caution}. The penalized efficient score function \citep{jaman2025penalized} for $\vec\theta=(\vec\delta^\top,\vec\psi^\top)^\top$ is
\begin{align*}
  \vec S^{P}(\vec\theta) = \vec S^{\text{eff}}(\vec\theta) - n{\vec q}_{\lambda_n}(|\vec\theta|)\text{sign}(\vec\theta),
\end{align*}
where ${\vec q}_{\lambda_n}(|\vec\theta|) = (\vec 0',{\vec q}_{\lambda_n}(|\vec\psi'|))'$. The penalized estimates of $\vec\theta$ are obtained by solving the following equations:
\begin{align} \label{penalized.score.equation}
   \vec S^{P}(\vec\theta)=\vec 0. 
\end{align}
To solve the equations in (\ref{penalized.score.equation}), \cite{jaman2025penalized} proposed an iterative procedure that combines G-estimation with the minorization-maximization (MM) algorithm \citep{hunter2005variable} to handle the nonconvex penalty, and considered the doubly-robust information criterion \citep{bian2024variable, moodie2023variable} for tuning parameter selection. \cite{jaman2025penalized} established the asymptotic properties of the penalized G-estimator and verified the double-robustness property via simulations.

\subsection{The problem with post-selection inference}\label{sec.problem}
Different models carry different interpretations of the parameters and answer different questions. %In linear ordinary least squares (OLS) regression, the coefficient of a predictor in a particular model is the ``average difference in the response approximated by that model" for a unit difference in the predictor, at the fixed levels of all other adjuster covariates in that particular model. %will be beneficial if we give the age-income example here
%For example, consider the relationship of age with the purchase likelihood of a high-tech gadget (this example is taken from \cite{berk2013valid}). It may seem surprising if we see younger people have lower purchase likelihood. But, if we look at the effect of age after adjusting for the income, we may find that at fixed levels of income younger people indeed have higher purchase likelihood. This paradox is a result of the positive collinearity between age and income. Hence, the income-unadjusted effect of age is different from income-adjusted effect of age, and two coefficients answer two different questions. For explanation in the SNMM framework, 
Let $\vec h_{ij}(M)$ denote the vector of observed covariates at the $j$-th measurement occasion for subject $i$ corresponding to the blip submodel $M$ and let $\vec\psi_M$ denote the target parameter vector under submodel $M$. The target of estimation using $\widehat{\gamma}_{j.M} = a_{ij}\vec h_{ij}'(M)\widehat{\vec\psi}_M$ given the blip submodel $M$ is $\gamma_{j.M} = a_{ij}\vec h_{ij}'(M)\vec\psi_M$% (for ease of expression we used $\gamma_{j.M}$ instead of $\gamma_{j.M}(a_{ij},\vec h_{ij}(M);\vec\psi_M)$)
. Therefore,  we do not unbiasedly estimate the true $\gamma_{j}^*$, rather we estimate its approximation $\gamma_{j.M}$ with respect to submodel $M$. If the submodel $M$ is subject to any kind of model selection using the observed data $\mathcal{D}$, then we should express the selected model as $\widehat{M} = \widehat{M}(\mathcal{D})$, which is now random. The selected model could be different for another realization $\mathcal{D}^*$. 
The target vector of coefficients $\vec\psi_{\widehat{M}(\mathcal{D})}$ for selected model $\widehat{M}(\mathcal{D})$ is also random: a) $\vec\psi_{\widehat{M}(\mathcal{D})}$ may have a different dimension for different data, b) a particular covariate may or may not be present in $\widehat{M}(\mathcal{D})$, and c) for any covariate in $\widehat{M}(\mathcal{D})$, its coefficient value may depend on the set of other covariates in $\widehat{M}(\mathcal{D})$. 
So, the set of parameters for which inference is sought is also random.% Consequently, we need an inference approach that accounts for all types of randomness involved in the estimation.

The naive post-selection inference procedure neither takes into account the uncertainty associated with model selection nor the possibility of selecting an incorrect model. Although the regularized procedure discussed in Section~\ref{sec.penalizedG} showed good performance in identifying effect modifiers with consistent estimation of the target parameters \citep{jaman2025penalized}, such a regularization method may yield estimators with distributions that are difficult to approximate. \cite{jaman2025penalized} proved the desirable asymptotic properties of the penalized G-estimator and presented a sandwich formula for calculating its asymptotic variance. Such sandwich estimators are consistent even when the number of parameters tends to infinity \citep{fan2004nonconcave}. If we have infinitely large samples, naive inference based on the sandwich variance estimator is valid. However, in finite samples, the uncertainty in the selection of effect modifiers invalidates the post-selection inference based on this sandwich variance. In practice, sandwich estimator tends to underestimate the standard errors, and the derived normal confidence regions (CRs) often do not provide acceptable coverage in finite samples \citep{minnier2011perturbation}. We have also observed the same issue with the sandwich estimator in our context (see our simulation results).
    
    For illustration, we consider a simple linear regression example. Suppose the data generating model is
    \begin{equation*}
    Y=\delta_1 X_1 + \delta_2 X_2 + \epsilon, \quad \epsilon \sim N(0, \sigma^2).
    \end{equation*}
    If we pre-specify our model with both covariates, then the OLS estimator for $\delta_1$ has the desirable properties and the standard $t$-statistic based inference is valid. However, if we use the same data to select variables-- for example, by including $X_2$ only if its sample correlation with $Y$ exceeds some threshold value-- then the distribution of the OLS estimator $\widehat{\delta}_1$ depends on whether $X_2$ was included. Standard errors and $p$-values obtained under the data-driven model no longer reflect this additional variability. Even if the selection method (e.g., penalization) yields estimators with favorable asymptotic properties, in finite samples the selection step can have a substantial influence on the distribution of the estimator. Hence, developing a valid post-selection inference method for the penalized G-estimation is crucial.

\subsection{Our proposals for valid inference with penalized G-estimation}
In this section, we present two different proposals for our target of inference: i) an extension of the random design UPoSI approach \citep{kuchibhotla2020valid} to the context of effect modification analysis in SNMMs using penalized G-estimation, and ii) an inference method based on a one-step improved penalized G-estimator derived from a decorrelated score function, following an approach similar to that of \cite{ning2017general}. Both the UPoSI approach \citep{kuchibhotla2020valid} and the decorrelated score method \citep{ning2017general} were originally developed for OLS regression.

\subsubsection{UPoSI approach}
Let $v \in \mathbb{R}^q$ is a vector of dimension $q$, $v(j)$ denotes the $j$-th element of $v$, and for any square-symmetric matrix $B \in \mathbb{R}^{q\times q}$, $B(j,k)$ denotes the element of $B$ corresponding to the $j$-th row and $k$-th column. We express the $r$-norm of the vector $v$ as
\begin{align*}
    ||v||_r=\Bigg(\sum_{j=1}^q|v(j)|^r\Bigg)^{1/r} \text{ for } 1 \leq r < \infty
\end{align*}
and the largest element in $v$ as
\begin{align*}
    ||v||_\infty = \underset{1 \leq j \leq q}{\text{max}}|v(j)|.
\end{align*}
Similarly, $||B||_\infty$ denotes the largest element (in absolute value) of the matrix $B$. We will use the term ``submodel" to refer to a subset of covariates in the regression and denote it by $M \subseteq \{1, 2, \ldots, K\}$. We define the set of all nonempty models of size no larger than $k$ by
\begin{align*}
    \mathcal{M}_K(k) = \{M : M \subseteq \{1,2,\ldots, K\}, 1 \leq |M| \leq k\}, \text{ for } 1 \leq k \leq K
\end{align*}
where $|M|$ represents the cardinality of $M$. Note that $\mathcal{M}_K(k)$ is the power set of $\{1,\ldots,K\}$ excluding the empty set. %%We will always include exposure in the model.. explain or edit the model universe addressing this 
The matrix $\vec h_i$ of dimension $J\times K$ contains the observed values of adjuster variables for subject $i$, and let $\vec h_i(M)$ denote the submatrix of $\vec h_i$  corresponding to submodel $M$. We also define the following quantities:
\begin{align*}
    \widehat{\vec W}_n = \frac{1}{n}\sum_{i=1}^n
    \begin{bmatrix}
    \vec h_i^\top\\
    [\{\vec a_i-\widehat{E}(\vec A_i|\vec H_i)\}\circledast\vec h_i]^\top
    \end{bmatrix}\widehat{\vec V}_i^{-1} \begin{bmatrix}
    \vec h_i & \vec a_i\circledast\vec h_i
    \end{bmatrix}
\end{align*}
and
\begin{align*}
    \widehat{\vec G}_n = \frac{1}{n}\sum_{i=1}^n\begin{bmatrix}
    \vec h_i^\top\\
    [\{\vec a_i-\widehat{E}(\vec A_i|\vec H_i)\}\circledast\vec h_i]^\top
    \end{bmatrix}\widehat{\vec V}_i^{-1}\vec Y_i,
\end{align*}
where $\vec a_i\circledast\vec h_i = \text{diag}(\vec a_i)\vec h_i$ denotes the row-wise multiplication, i.e., each row of $\vec h_i$ is multiplied by the corresponding element of $\vec a_i$, $\vec V_i$ is the shorthand notation for $Var(\vec U_i|\vec H_i)$, and $\widehat{E}(\vec A_i|\vec H_i)$ is the $J$-dimensional vector of estimated propensity scores. Note that $\widehat{\vec W}_n$ is a $(2K\times 2K)$-dimensional matrix and $\widehat{\vec G}_n$ is a $(2K\times 1)$-dimensional vector, and these are defined in terms of the full model, which refers to $M$ where $|M| = K$. We denote the expected versions of these quantities by $\vec W_n$ and $\vec G_n$, and define the estimation errors of $\vec W_n$ and $\vec G_n$ as follows:
\begin{align*}
    D_n^{W} &= ||\widehat{\vec W}_n - \vec W_n||_\infty = \underset{M \in \mathcal{M}_K(2)}{\text{max}}||\widehat{\vec W}_n(M) - \vec W_n(M)||_\infty\\
    D_n^{G} &= ||\widehat{\vec G}_n - \vec G_n||_\infty = \underset{M \in \mathcal{M}_K(1)}{\text{max}}||\widehat{\vec G}_n(M) - \vec G_n(M)||_\infty,
\end{align*}
where $\mathcal{M}_K(2)$ and $\mathcal{M}_K(1)$ represent the sets of all models of sizes bounded by 2 and 1, respectively. Also note that $\widehat{\vec W}_n(M)$ is the submatrix of $\widehat{\vec W}_n$ and $\widehat{\vec G}_n(M)$ is the subvector of $\widehat{\vec G}_n$ corresponding to submodel $M$, which are defined as follows:
\begin{align} \label{eq.W_nM}
    \widehat{\vec W}_n(M) = \frac{1}{n}\sum_{i=1}^n
    \begin{bmatrix}
    \vec h_i^\top\\
    [\{\vec a_i-\widehat{E}(\vec A_i|\vec H_i)\}\circledast\vec h_i(M)]^\top
    \end{bmatrix}\widehat{\vec V}_i^{-1}\begin{bmatrix}
    \vec h_i & \vec a_i\circledast\vec h_i(M)
    \end{bmatrix}
\end{align}
and
\begin{align}
    \widehat{\vec G}_n(M) = \frac{1}{n}\sum_{i=1}^n\begin{bmatrix}
    \vec h_i^\top\\
    [\{\vec a_i-\widehat{E}(\vec A_i|\vec H_i)\}\circledast\vec h_i(M)]^\top
    \end{bmatrix}\widehat{\vec V}_i^{-1}\vec Y_i.
\end{align}
%
% Need to show the derivation of normal1 in article
The empirical and the expected versions of the unpenalized estimating equations  $\vec S^{eff}(\vec\theta)=\vec 0$ corresponding to the submodel $M$ can be written as
\begin{align}
    \widehat{\vec W}_n(M)\widehat{\vec\theta}_{n,M}&=\widehat{\vec G}_n(M)\quad\text{and}\label{normal1}\\
    \vec W_n(M)\vec\theta_{n,M}&=\vec G_n(M),\label{normal2}
\end{align}
where $\widehat{\vec\theta}_{n,M}$ denotes the G-estimator of $\vec\theta_M$, the target parameters under submodel $M$. We propose the following UPoSI confidence regions for the G-estimator $\widehat{\vec\theta}_{n,M}$ in the selected submodel $M$:
\begin{align} \label{uposi.region.asymp}
    \widehat{\mathcal{R}}_{n,M}^*:=\bigg\{\vec\theta \in \mathbb{R}^{|M|}:||\widehat{\vec W}_n(M)\{\widehat{\vec\theta}_{n,M}-\vec\theta\}||_\infty \leq C_n^{G}(\alpha)+ C_n^{W}(\alpha)||\widehat{\vec\theta}_{n,M}||_1\bigg\},
\end{align}
where $C_n^{G}(\alpha)$ and $C_n^{W}(\alpha)$ are the bivariate joint upper $\alpha$ quantiles of $D_n^{G}$ and $D_n^{W}$, by construction satisfying
\begin{align*}
    P\bigg(D_n^{G} \leq C_n^{G}(\alpha) \text{ and } D_n^{W} \leq C_n^{W}(\alpha)\bigg) \geq 1-\alpha.
\end{align*}
The regions in (\ref{uposi.region.asymp}) provide an asymptotic coverage guarantee. The region that provides a finite sample guarantee can be defined as% \cite{kuchibhotla2020valid}
\begin{align} \label{uposi.region.finite}
    \widehat{\mathcal{R}}_{n,M}:=\bigg\{\vec\theta \in \mathbb{R}^{|M|}:||\widehat{\vec W}_n(M)\{\widehat{\vec\theta}_{n,M}-\vec\theta\}||_\infty \leq C_n^{G}(\alpha)+ C_n^{W}(\alpha)||\vec\theta||_1\bigg\},
\end{align}
The regions in (\ref{uposi.region.finite}) can be obtained by doing simple mathematical operations on equations~\ref{normal1} and~\ref{normal2} (see Theorem~\ref{thm.equiv} and Theorem~\ref{thm.UPoSI.region.derivation} and the proofs in Appendix~\ref{A.tech.detail.uposi}). Since the regions in (\ref{uposi.region.finite}) are difficult to analyze in terms of shape and Lebesgue measure \cite{kuchibhotla2020valid}, we focus on constructing the regions in (\ref{uposi.region.asymp}) for the penalized G-estimator. The quantiles $C_n^{G}(\alpha)$ and $C_n^{W}(\alpha)$ are unknown, which must be estimated from the data using a bootstrap method. In our study, we use the multiplier bootstrap for estimating the joint quantiles, which is described in Appendix~\ref{A.sec.multp.boot}. Also, we can construct  coordinate-wise confidence intervals like the form shown in Appendix~\ref{A.sec.conf.int}. Asymptotic validity of the  UPoSI method is described in Theorem~\ref{thm.asymp.validity.UPoSI} and the proof is given in Appendix~\ref{A.tech.detail.uposi}.

\begin{theorem}[\bf Asymptotic validity of UPoSI]\label{thm.asymp.validity.UPoSI}
 Let $\lambda_{\text{min}}(\vec W_n(M))$ denote the minimum eigen value of the matrix $\vec W_n(M)$. For every $1 \leq k \leq K$ satisfying the assumption that the estimation error $D_n^{W}$ satisfies $kD_n^{W} = o_{\mathbb{P}}(\omega_n(k))$ as $n \rightarrow \infty$, where $\omega_n(k)=\text{min}_{M \in \mathcal{M}_K(k)}\lambda_{\text{min}}(\vec W_n(M))$, the confidence regions $\widehat{\mathcal{R}}_{n,M}^*$ in (\ref{uposi.region.asymp}) satisfy
 \begin{align*}
    \underset{n \rightarrow \infty}{\text{lim inf }}     P\Bigg(\bigcap_{M \in \mathcal{M}_K(k)}\{\vec\theta_{n,M} \in \widehat{\mathcal{R}}_{n,M}^*\}\Bigg) \geq 1 - \alpha.
 \end{align*}
\end{theorem}
For a specific correlation structure (corstr), the steps for the whole estimation procedure are summarized in Algorithm~\ref{UPoSI.algorithm}.

\begin{algorithm}
\caption{{\bf U}niversal {\bf Po}st-{\bf S}election {\bf I}nference for {\bf Pe}nalized {\bf G}-estimation}\label{UPoSI.algorithm}
\begin{algorithmic}[1]
\Procedure{UPoSIPeG}{$\vec A, \vec H, \vec Y, \widehat{\vec\theta}, \widehat{\sigma}, \widehat{\rho}$,\,corstr,\,$\alpha,\widehat{M}$}
    \State Compute $\text{E}(\vec A_i|\vec H_i)$ for $i=1,\ldots,n$;
    \State Compute $\widehat{\vec V}_i$ using $\widehat{\sigma}$ and $\widehat{\rho}$ according to the corstr for $i=1,\ldots,n$;
    \State Compute $\widehat{\vec W}_n(\widehat{M})$ following (\ref{eq.W_nM});
    \State Standardize the continuous variables in $\vec H$;
    \State \parbox[t]{0.92\linewidth}{Obtain the bivariate quantiles $\widehat{C}_n^G$ and $\widehat{C}_n^W$ following the multiplier bootstrap method described in Appendix~\ref{A.sec.multp.boot};}
    \State Using $\widehat{\vec\theta}=(\widehat{\vec\delta}, \widehat{\vec\psi})^\top$ define $\mathcal{B} = \{0\} \ltxcup \{m: m\in \{1,\ldots, K-1\} \text{ and } |\widehat{\psi}_m| \geq 0.001\}$;
    \For {each $k \in \mathcal{B}$}
        \State \parbox[t]{0.92\linewidth}{Construct the $(1-\alpha)\times 100\%$ confidence interval for the $k$-th coefficient in $\vec\psi$ as
        \begin{align*}
            \widehat{\psi}_k \pm \Big|\vec c_k'\big\{\widehat{\vec W}_n(\widehat{M})\big\}^{-1}\Big| \Big(\widehat{C}_n^G(\alpha) + \widehat{C}_n^W(\alpha)||\widehat{\vec\theta}||_1\Big),
        \end{align*}
        where $\vec c_k$ is a vector that contains 1 at the $k$-th position and zeros elsewhere.}
    \EndFor
    \State \textbf{return} the confidence intervals for the blip coefficients $\psi_k$, where $k \in \mathcal{B}$.
\EndProcedure
\end{algorithmic}
\end{algorithm}

\subsubsection{One-step improved penalized G-estimator%or Decorrelated score test
}
%Applying Taylor expansion of the score function, \cite{ning2017general} showed that the asymptotic normality of the score test statistic fails due to the non-ignorable estimation bias and sparsity effect of the high dimensional parameters in the higher order terms of the Taylor approximation. \cite{ning2017general} proposed debiasing the score function of the parameter of interest by de-correlating it from the high dimensional nuisance scores. (Does it provide simultaneous inference?)
For valid inference about the target parameter $\vec\psi$, we can derive a one-step improvement \citep{ning2017general,gao2025asymptotic} of the penalized-G estimator. We make a partition of the target parameter vector as $\vec\psi=(\psi_k,\vec\nu_k)$, where  $k$ can take any value in $\{0,1,\ldots, K-1\}$ and $\vec\nu_k = (\psi_0,\ldots,\psi_{k-1},\psi_{k+1},\ldots, \psi_{K-1})$. Let $\vec S(\vec\theta)= n^{-1}\vec S^{\text{eff}}(\vec\theta)$, where $\vec S^{\text{eff}}(\vec\theta)= \sum_{i=1}^n \vec S_i^{\text{eff}}(\vec\theta)$. Suppose $\vec S_{\vec\psi} = (S_{\psi_k}, \vec S_{\vec\nu_k}^\top)^\top$ denote the sub-vector of $\vec S(\vec\theta)$ corresponding to the parameters in $\vec\psi$ and $\vec I=E[\vec S_{\vec\psi}\vec S_{\vec\psi}^\top]$. We denote  the  submatrix of $\vec I$ corresponding to the parameters in $\vec\psi$ by $I_{\psi_k\psi_k}$, $\vec I_{\psi_k\vec\nu_k}$, $\vec I_{\vec\nu_k\vec\nu_k}$, $\vec I_{\vec\nu_k\psi_k}$ and define $I_{\psi_k|\vec\nu_k} = I_{\psi_k\psi_k} - \vec I_{\psi_k\vec\nu_k}\vec I_{\vec\nu_k\vec\nu_k}^{-1}\vec I_{\vec\nu_k\psi_k}$. A decorrelated score function can be defined as
\begin{align} \label{eq.decorr.score.fun}
    \ddot{S}(\psi_k,\vec\nu_k,\vec\delta) = S_{\psi_k} - \vec w^\top\vec S_{\vec\nu_k},
\end{align}
where $\vec w^\top=\vec I_{\psi_k\vec\nu_k}\vec I_{\vec\nu_k\vec\nu_k}^{-1}$. The score function $\ddot{S}(\psi_k,\vec\nu_k,\vec\delta)$ is uncorrelated with the nuisance score function $\vec S_{\vec\nu_k}$ in the sense that
\begin{align*}
    E[\ddot{S}(\psi_k,\vec\nu_k,\vec\delta)\vec S_{\vec\nu_k}] = \vec 0.
\end{align*}
The decorrelation operation controls the variability of higher order terms in the Taylor expansions of the score function $\vec S(\vec\theta)$. We need to impose some sparsity assumption on $\vec w$ to control the estimation error, i.e., we will find the estimator $\widehat{\vec w}$ of $\vec w$ that searches for the best sparse linear combination of the nuisance score functions to approximate the score function of the parameter of interest.
%Given the penalized G-estimator $\widehat{\vec\theta}=(\widehat{\vec\delta}^\top, \widehat{\vec\psi}^\top)^\top$, the estimator $\widehat{\vec w}$ of $\vec w$ can be obtained by any sparse estimation method, for example, LASSO, the Dantzig selector etc. %Let $\widehat{\vec v} = (1, -\widehat{\vec w}^\top)^\top$. 
We plug-in the estimates of the treatment-free model parameters, and estimate $\ddot{S}(\psi_k,\vec\nu_k,\widehat{\vec\delta})$ as follows:
\begin{align} \label{eq.decorr.score.est}
    \widehat{\ddot{S}}(\widehat{\psi}_k,\widehat{\vec\nu}_k,\widehat{\vec\delta}) = \widehat{S}_{\psi_k} - \widehat{\vec w}^\top\widehat{\vec S}_{\vec\nu_k},
\end{align}
which can be used for hypothesis testing.

The decorrelated score function can be regarded as an approximately unbiased estimating function for $\psi_k$ and an estimator of $\psi_k$ can be found by solving $\widehat{\ddot{S}}(\psi_k,\widehat{\vec\nu}_k,\widehat{\vec\delta}) = 0$. Since $\widehat{\ddot{S}}(\psi_k,\widehat{\vec\nu}_k,\widehat{\vec\delta})$ may have multiple roots, we can find an estimator by solving the first order approximation of $\widehat{\ddot{S}}(\psi_k,\widehat{\vec\nu}_k,\widehat{\vec\delta}) = 0$. Given the sparse estimator $\widehat{\vec\psi}$ and the estimated partial information
\begin{align}\label{part.inf}
    \widehat{I}_{\psi_k|\vec\nu_k} = \widehat{I}_{\psi_k\psi_k} - \widehat{\vec w}^\top\widehat{\vec I}_{\vec\nu_k\psi_k},
\end{align}
where $\widehat{I}_{\psi_k\psi_k} = \widehat{S}_{\psi_k}^2$ and $\widehat{\vec I}_{\vec\nu_k\psi_k} = \widehat{\vec S}_{\vec\nu_k}\widehat{S}_{\psi_k}$, we find the one-step improved penalized G-estimator $\widetilde{\psi}_k$ of $\psi_k$ by solving $\widehat{\ddot{S}}(\widehat{\psi}_k,\widehat{\vec\nu}_k,\widehat{\vec\delta}) + \widehat{I}_{\psi_k|\vec\nu_k}(\psi_k - \widehat{\psi}_k) = 0$ and the solution is as follows
\begin{align} \label{eq.os.est.part.inform}
    \widetilde{\psi}_k = \widehat{\psi}_k - \widehat{\ddot{S}}(\widehat{\psi}_k,\widehat{\vec\nu}_k,\widehat{\vec\delta})/\widehat{I}_{\psi_k|\vec\nu_k}.
\end{align}
We show the asymptotic normality of the decorrelated score function in Theorem~\ref{thm.asymp.normal.decorr.score} in Appendix~\ref{A.tech.detail.one.step}. The asymptotic normality of the one-step improved penalized G-estimator $\widetilde{\psi}_k$ is stated in Theorem~\ref{thm.asymp.normal.OS} and is proved using the result of Theorem~\ref{thm.asymp.normal.decorr.score} (see Appendix~\ref{A.tech.detail.one.step} for the proof).

\begin{theorem}[\bf Asymptotic normality of the one-step improved penalized G-estimator]\label{thm.asymp.normal.OS}
    Under the regularity conditions C1-C6 and the Assumptions \ref{asmp1}-\ref{asmp8} mentioned in Appendix~\ref{A.tech.detail.one.step}, if $\{\eta_1(n) + \eta_2(n)\}\sqrt{\log K} = o(1)$, $\widehat{I}_{\psi_k|\vec\nu_k}$ is consistent for $ I_{\psi_k|\vec\nu_k}^*$, and $I_{\psi_k|\vec\nu_k}^* \geq C$ for some constant $C > 0$, then
    \begin{align*}
        n^{1/2}(\widetilde{\psi}_k - \psi_k^*)I_{\psi_k|\vec\nu_k}^*/\sigma_S^{*1/2} = -\vec S_{\vec\psi^*}/\sigma_S^{*1/2} + o_{\mathbb{P}}(1) \sim N(0,1)
    \end{align*}
    for $k = 0, 1, \ldots, K-1$, where $\sigma_S^*$ is defined in Assumption~\ref{asmp8} in Appendix~\ref{A.tech.detail.one.step}.
\end{theorem}

Based on the results of Theorem~\ref{thm.asymp.normal.OS}, we can construct a $(1-\alpha)\times 100\%$ confidence interval of $\psi_k$ as
\begin{align} \label{eq.CI.os.est}
    \Big(\widetilde{\psi}_k - \Phi^{-1}(1-\alpha/2) \frac{\sqrt{\widehat{\sigma}_S}}{\sqrt{n}\widehat{I}_{\psi_k|\vec\nu_k}}, \widetilde{\psi}_k + \Phi^{-1}(1-\alpha/2) \frac{\sqrt{\widehat{\sigma}_S}}{\sqrt{n}\widehat{I}_{\psi_k|\vec\nu_k}}\Big),
\end{align}
where $\Phi$ denotes the cumulative distribution function (CDF) of standard normal distribution, and 
\begin{align} \label{eq.sigma_s}
    \widehat{\sigma}_S= (1, -\widehat{\vec w}^\top)\widehat{\vec S}_{\vec\psi}\widehat{\vec S}_{\vec\psi}^\top(1, -\widehat{\vec w}^\top)^\top.
\end{align}
For a specific correlation structure (corstr), the steps for the whole estimation procedure are summarized in Algorithm~\ref{OS.algorithm}.

\begin{algorithm}
\caption{Inference Using the {\bf O}ne-{\bf S}tep {\bf I}mproved {\bf Pe}nalized {\bf G}-estimator}\label{OS.algorithm}
\begin{algorithmic}[1]
\Procedure{OSIPeG}{$\vec A, \vec H, \vec Y, \widehat{\vec\theta}, \widehat{\sigma}, \widehat{\rho}$,\,corstr,\,$\alpha, \vec\lambda_{\text{seq}}$}
\State Standardize the continuous variables in $\vec H$;
\State Compute $\widehat{\vec V}_i$ using $\widehat{\sigma}$ and $\widehat{\rho}$ according to the corstr for $i=1,\ldots,n$;
\State $\vec e_i \gets \vec Y_i - (\vec H_i\;\;\vec A_i\circledast\vec H_i)\widehat{\vec\theta}$ for $i=1,\ldots,n$;
\State $\vec S_{\vec\psi} \gets n^{-1}\sum_{i=1}^n \vec S_{\vec\psi,i}$, where $\vec S_{\vec\psi,i} = [\{\vec A_i - E(\vec A_i|\vec H_i)\}\circledast\vec H_i]^\top\widehat{\vec V}_i^{-1}\vec e_i$;
\State $\vec I_{\vec\psi} \gets n^{-1}\sum_{i=1}^n \vec S_{\vec\psi,i}\vec S_{\vec\psi,i}^\top$;
\State Using $\widehat{\vec\theta}=(\widehat{\vec\delta}, \widehat{\vec\psi})^\top$ define $\mathcal{B} = \{0\} \ltxcup \{m: m\in \{1,\ldots, K-1\} \text{ and } |\widehat{\psi}_m| \geq 0.001\}$;
    \For{each $k \in \mathcal{B}$}
        \State Partition the target estimate $\widehat{\vec\psi}$ as $(\widehat{\psi}_k, \widehat{\vec\nu}_k)$;
        \For{each $\lambda_w \in \vec\lambda_{\text{seq}}$}
            \State Obtain the Dantzig type estimator $\widehat{\vec w}_{\lambda_w}$ according to (\ref{eq.dantzig});
        \EndFor
            \State Choose the optimal $\lambda_w^*$ using cross-validation and set $\widehat{\vec w} = \widehat{\vec w}_{\lambda_w^*}$;
            \State Compute the decorrelated score function $\widehat{\ddot{S}}(\widehat{\psi}_k,\widehat{\vec\nu}_k,\widehat{\vec\delta})$ using (\ref{eq.decorr.score.est});
            \State Compute $\widehat{I}_{\psi_k|\vec\nu_k}$ according to (\ref{part.inf});
            \State Calculate the one-step improved estimator $ \widetilde{\psi}_k$ using (\ref{eq.os.est.part.inform});
            \State \parbox[t]{0.92\linewidth}{Construct the $(1-\alpha)\times 100\%$ confidence interval for $\psi_k$ following (\ref{eq.CI.os.est});}
    \EndFor
    \State \textbf{return} $\widetilde{\vec\psi}$ and the confidence intervals for the selected coefficients.
\EndProcedure
\end{algorithmic}
\end{algorithm}

In our study, we evaluate the performance of the method with two different sparse weight estimators:
\begin{align}
 \text{\textbf{LASSO}: }   \widehat{\vec w} &= \underset{\vec w}{\text{argmin}} \frac{1}{2n}\lVert \widehat{S}_{\psi_k} - \vec w^\top \widehat{\vec S}_{\vec\nu_k}\rVert_2^2 + \lambda_{\vec w} \lVert\vec w\rVert_1\\
  \text{\textbf{Dantzig selector}: }  \widehat{\vec w} &= \underset{\vec w}{\text{argmin}} \lVert\vec w\rVert_1 \text{\hspace{0.5cm}s.t. }\lVert \widehat{\vec S}_{\vec\nu_k}^\top(\widehat{S}_{\psi_k} - \vec w^\top \widehat{\vec S}_{\vec\nu_k})\rVert_\infty \leq \lambda_{\vec w} \label{eq.dantzig}
\end{align}
where $\lambda_{\vec w}$ is the tuning parameter, which we choose by cross-validation. LASSO and the Dantzig selector both produce sparse weight estimates. While LASSO minimizes the residual sum of squares with an $L_1$-penalty on the weights, the Dantzig selector minimizes $L_1$-norm of the weights with a constraint on the maximum absolute correlation between residuals and nuisance scores. LASSO is computationally faster and performs well when nuisance scores are not highly correlated. On the other hand, the Dantzig selector is robust and more stable in scenarios where the nuisance scores are highly correlated, but it is computationally intensive.

\section{Simulation study} \label{sec.simulation}

We use a simulation setting similar to that of \cite{jaman2025penalized}. To generate the data for the $j$-th session ($j=1,\ldots,J$), % of each subject, 
we generated two baseline confounders as $L^{(1)} \sim N(0,1)$ and $L^{(2)} \sim N(0,1)$, and the time varying confounders and noise covariates as $L_j^{(3)},\ldots,L_j^{(6)},X_j^{(1)},\ldots,X_j^{(K-6)} \sim MVN_{K-2}\Big((\vec\mu_{L,j}^\top,\vec\mu_{X,j}^\top)^\top,\vec V_{LX}\Big)$, where $\mu_{L,j}^{(k)}= 0.3\,l^{(k)}_{j-1}+0.3\,a_{j-1}$ for $k=3,4, 5 \text{ and } 6$, and $\mu_{X,j}^{(r)}= 0.5\,x^{(r)}_{j-1}$ for $ r=1,\ldots,K-6$. The covariance matrix $\vec V_{LX}$ has $(r,s)$-th element equal to $\tau^{|r-s|}$ for $r,s=1,\ldots,K-2$. We generated the binary exposure according to the probability
\begin{align}
    \mathbb{P}(A_j = 1 | \vec H_j) = \dfrac{\exp{\{\beta_0+\beta_1 l^{(1)}+\beta_2 l^{(2)}+\sum_{m=3}^6 \beta_m l_j^{(m)}\}}}{1+\exp{\{\beta_0+\beta_1 l^{(1)}+\beta_2 l^{(2)}+\sum_{m=3}^6 \beta_m l_j^{(m)}\}}}.
\end{align}
We then generated a vector of correlated errors $\vec \epsilon \sim N_{J}(\vec 0, \vec \Sigma)$, where $\vec \Sigma=\sigma^2_\epsilon\vec R$ is the variance-covariance matrix and $\vec R$ is the $J\times J$ correlation matrix defined with parameter $\rho$ according to an ``exchangeable" correlation structure, i.e., $\text{Corr}(\epsilon_{ij}, \epsilon_{ij}) = 1$ and $\text{Corr}(\epsilon_{ij}, \epsilon_{ij'}) = \rho$ for $j\neq j'$. We constructed the outcome as $y_j = \mu_j(\vec h_j; \vec\delta) + \gamma^*_j(a_j,\vec h_j; \vec\psi) + \epsilon_j$, where
\begin{align*}
    \mu_j(\vec h_j; \vec\delta)&=\delta_0+\delta_1 l^{(1)}+\delta_2 l^{(2)}
    +\sum_{m=3}^6 \delta_m l_j^{(m)}
    +\sum_{m=1}^{20} \delta_{6+m}\,x_j^{(m)}
    +\sum_{m=21}^{K-6} \delta_{6+m}\,x_j^{(m)}\\
    &\hspace{0.2cm}
    +\delta_{K+1} l^{(1)}l_j^{(5)}
    +\delta_{K+2} l_j^{(3)}l_j^{(4)}
    +\delta_{K+3} \sin(l_j^{(3)} - l_j^{(4)})
    +\delta_{K+4} \cos(2l_j^{(5)})    
\end{align*}
is the true treatment-free model and 
$\gamma^*_j(a_j,\vec h_j; \vec\psi)=(\psi_0+\psi_1 l^{(1)}+\psi_2 l^{(2)}+\sum_{m=3}^6\psi_m l_j^{(m)}+\sum_{m=1}^{20} \psi_{6+m}\,x_j^{(m)}+\sum_{m=21}^{K-6} \psi_{6+m}\,x_j^{(m)})a_j$ is the true blip function, with common parameters at each time point.
Let $\vec\beta=(\beta_0,\ldots,\beta_6)^\top$, $\vec\delta=(\delta_0,\ldots,\delta_{K+4})^\top$ and $\vec\psi=(\psi_0,\ldots,\psi_{K})^\top$. We set
\begin{align*}
   \vec\beta&=(0,1,-1.1,1.2,0.75,-0.9,1.2)^\top \\
   \vec\delta&=(1,1, 1.2, 1.2, -0.9, 0.8, -1, 1, \ldots, 1, 0, \ldots,0,-0.8, 1, 1.2, -1.5)^\top\\
   \vec\psi&=(1, 1, -1, -0.9, 0.8, 1, 0, 0, \ldots,0, 0, \ldots,0)^\top
\end{align*}
Note that $X^{(1)}$ to $X^{(20)}$ have impact on the outcome only and the coefficients of $X^{(21)}$ to $X^{(K-6)}$ were set to zero in the treatment-free model $\mu_j(\vec h_j; \vec\delta)$. The outcome-predictor $X^{(10)}$ was treated as unmeasured. Though we set the coefficients of all the $X$'s to zero in the blip function $\gamma^*_j(a_j,\vec h_j; \vec\psi)$, we consider a scenario where there is interest in investigating effect heterogeneity by the $X$ variables in addition to the $L$ variables. We consider $K=$ 20, 50 and 100, $n=$ 500, 800 and 1200, $\tau=0.3$, $\sigma^2_\epsilon=1$, and $\rho=0.8$. 

The linear outcome model used in the penalized estimation is misspecified because it ignores the non-linear terms in the treatment-free model and excludes the outcome-predictor $X^{(10)}$.
%\begin{itemize}
%    \item $L^{(1)}\times L^{(4)}$ and $L^{(2)}\times L^{(3)}$ interactions are excluded,
%    \item $\sin(L^{(3)}-L^{(4)})$ and $\cos(2 L^{(5)})$ terms were ignored, and
%    \item Covariate $X^{(10)}$ which also affects the outcome was treated as unmeasured.
%\end{itemize}
%\noindent
We have five competing methods for our target inference; a) Naive: The naive inference based on the Wald-type confidence interval constructed using the sandwich variance of the penalized G-estimator, b) UPoSI: Inference following the random design UPoSI approach, c) OS.FULL: Inference with the one step improved penalized G-estimator, where improvement is done using the full weight vector, d) OS.LASSO: Inference with the one step improved penalized G-estimator, where improvement is done using the sparse weight vector estimated by LASSO, and e) OS.Dantzig: Inference with the one step improved penalized G-estimator, where improvement is done using the sparse weight vector obtained by the Dantzig selector. For comparing the performance of the competing inferential methods, we use the following metrics: a) average confidence interval (CI) length of the blip coefficients in the selected model, b) false coverage rate in the selected model, and c) conditional power for the true non-zero blip coefficients in the selected model. Since the naive inference and the UPoSI approach are intended to provide interval estimates only for the selected coefficients, we assessed the conditional power, which refers to the likelihood that the confidence interval excludes zero for a true non-zero coefficients, given that it is included in the selected model.
%Definition of the metrics
We calculated the average CI length as $E[\sum_{k=1}^{\text{dim}(\vec\psi^*_{\widehat{M}})}(UL_k - LL_k)/\text{dim}(\vec\psi^*_{\widehat{M}})]$, the false coverage rate (FCR) as
\begin{align*}
    \text{FCR} = E\Bigg[\frac{\#\big\{1\leq k \leq \text{dim}(\vec\psi^*_{\widehat{M}}):(\vec\psi^*_{\widehat{M}})_k \not\in [LL_k,UL_k]\big\}}{\text{dim}(\vec\psi^*_{\widehat{M}})}\Bigg]
\end{align*}
and the conditional power as
\begin{align*}
    \text{Power} = E\Bigg[\frac{\#\big\{1\leq k \leq d_{\psi^*\neq 0}: 0 \not\in [LL_k,UL_k]\big\}}{d_{\psi^*\neq 0}}\Bigg],
\end{align*}
where $\widehat{M}$ denotes the selected model, $\vec\psi^*_{\widehat{M}}$ denotes the sub-vector of true blip parameters $\vec\psi^*$ according to $\widehat{M}$, $LL_k$ and $UL_k$ denote the lower and the upper confidence limits, respectively, and $d_{\psi^*\neq 0}$ denotes the number of true non-zero values in $\vec\psi^*_{\widehat{M}}$. The methods were evaluated under three different correlation structures: independent, exchangeable and unstructured. For the independent structure we used $\text{Corr}(U_{ij},U_{ij}) = 1$ and  $\text{Corr}(U_{ij},U_{ij'}) = 0$ for $j\neq j'$; for exchangeable, we used $\text{Corr}(U_{ij},U_{ij'}) = \rho$;  and for unstructured, $\text{Corr}(U_{ij},U_{ij'}) = \rho_{jj'}$. These correlation parameters were estimated using the residual-based moment method (details can be found in the supplementary material in \cite{jaman2025penalized}). The performance metrics were calculated from 150 independent simulations.

\begin{figure}[t]
    \centering
    \includegraphics[scale=0.8]{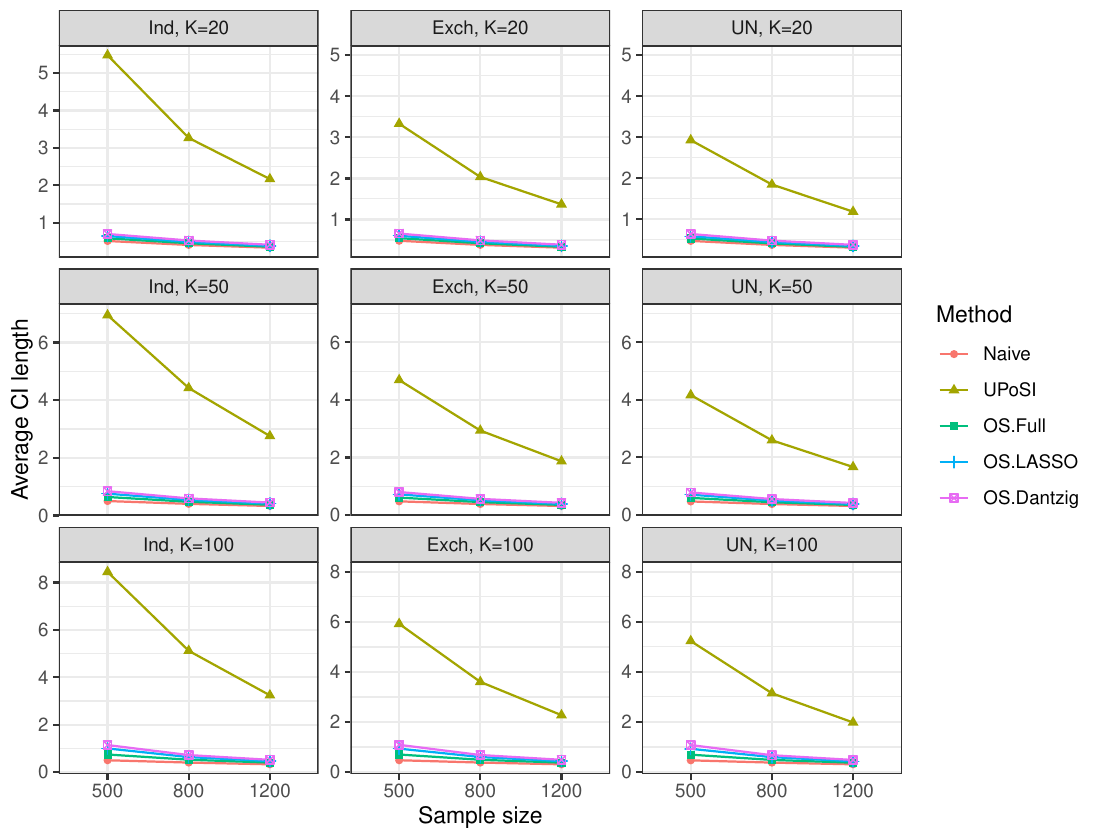}
    \caption{Average confidence interval lengths for each inferential method by number of covariates ($K$), sample size, and correlation structure (Ind: Independent, Exch: Exchangeable, UN: Unstructured).}
    \label{fig.averageCI}
\end{figure}

\begin{figure}[t]
    \centering
    \includegraphics[scale=0.8]{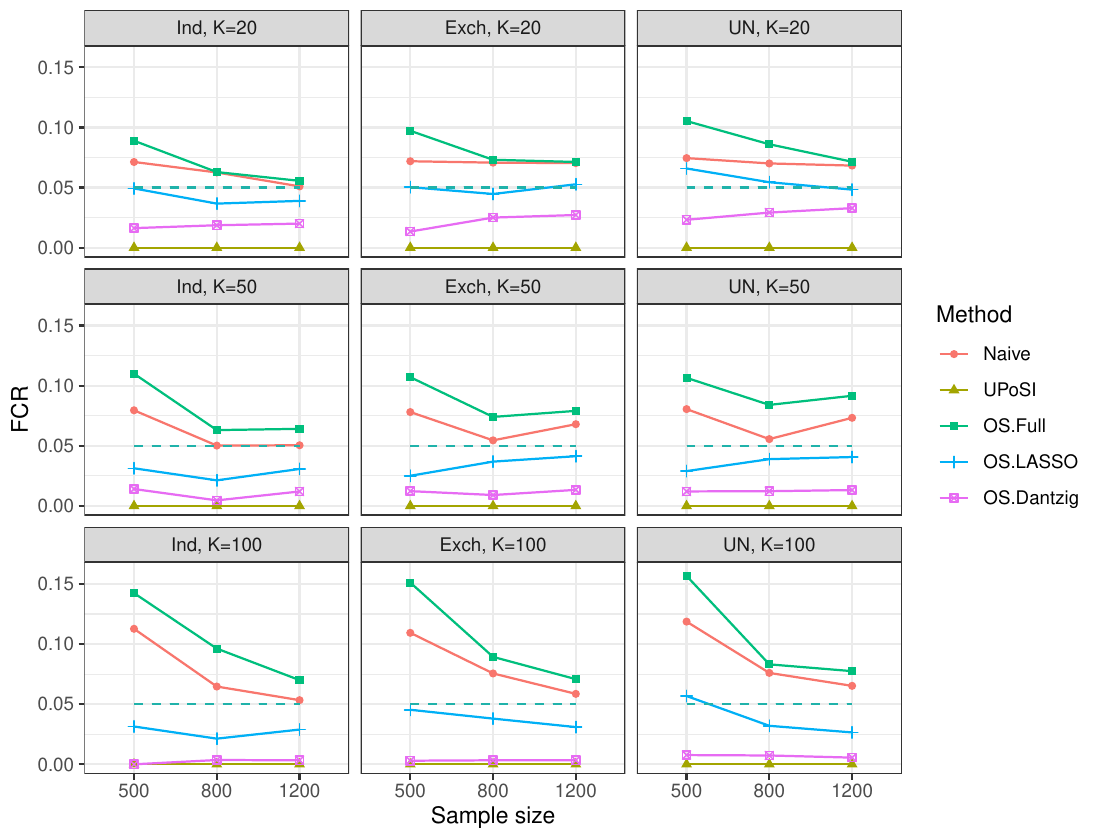}
    \caption{False coverage rates for each inferential method by number of covariates ($K$), sample size, and correlation structure (Ind: Independent, Exch: Exchangeable, UN: Unstructured).}
    \label{fig.fcr}
\end{figure}

\begin{table}[t]
\centering
\caption{Inferential power (for the selected non-zero coefficients) of different methods.} \label{tab.power}
\resizebox{1\textwidth}{!}{
  \begin{tabular}{clccccccccccc}
  \toprule
  &&\multicolumn{3}{c}{$K=20$}&&\multicolumn{3}{c}{$K=50$}&&\multicolumn{3}{c}{$K=100$}\\
  \cline{3-5}\cline{7-9}\cline{11-13}
  Sample size & Method &Ind&Exch&UN&&Ind&Exch&UN&&Ind&Exch&UN\\
  \hline
  $n=500$& Naive & 1.00 & 1.00 & 1.00 &  & 1.00 & 1.00 & 1.00 &  & 1.00 & 1.00 & 1.00 \\ 
   & UPoSI & 0.00 & 0.04 & 0.08 &  & 0.00 & 0.00 & 0.00 &  & 0.00 & 0.00 & 0.00 \\ 
   & OS.Full & 1.00 & 1.00 & 1.00 &  & 1.00 & 1.00 & 1.00 &  & 1.00 & 1.00 & 1.00 \\ 
   & OS.LASSO & 1.00 & 1.00 & 1.00 &  & 1.00 & 1.00 & 1.00 &  & 0.99 & 1.00 & 1.00 \\ 
   & OS.Dantzig & 1.00 & 1.00 & 1.00 &  & 1.00 & 1.00 & 1.00 &  & 0.99 & 0.99 & 1.00  \vspace{0.1cm}\\ 
  $n=800$& Naive & 1.00 & 1.00 & 1.00 &  & 1.00 & 1.00 & 1.00 &  & 1.00 & 1.00 & 1.00 \\ 
   & UPoSI & 0.04 & 0.39 & 0.53 &  & 0.00 & 0.06 & 0.13 &  & 0.00 & 0.00 & 0.02 \\ 
   & OS.Full & 1.00 & 1.00 & 1.00 &  & 1.00 & 1.00 & 1.00 &  & 1.00 & 1.00 & 1.00 \\ 
   & OS.LASSO & 1.00 & 1.00 & 1.00 &  & 1.00 & 1.00 & 1.00 &  & 1.00 & 1.00 & 1.00 \\ 
   & OS.Dantzig & 1.00 & 1.00 & 1.00 &  & 1.00 & 1.00 & 1.00 &  & 1.00 & 1.00 & 1.00   \vspace{0.1cm}\\ 
  $n=1200$& Naive & 1.00 & 1.00 & 1.00 &  & 1.00 & 1.00 & 1.00 &  & 1.00 & 1.00 & 1.00 \\ 
   & UPoSI & 0.32 & 0.94 & 0.98 &  & 0.11 & 0.52 & 0.73 &  & 0.02 & 0.20 & 0.40 \\ 
   & OS.Full & 1.00 & 1.00 & 1.00 &  & 1.00 & 1.00 & 1.00 &  & 1.00 & 1.00 & 1.00 \\ 
   & OS.LASSO & 1.00 & 1.00 & 1.00 &  & 1.00 & 1.00 & 1.00 &  & 1.00 & 1.00 & 1.00 \\ 
   & OS.Dantzig & 1.00 & 1.00 & 1.00 &  & 1.00 & 1.00 & 1.00 &  & 1.00 & 1.00 & 1.00  \\  
  \bottomrule
  \end{tabular}
}\\
Ind: Independent, Exch: Exchangeable, UN: Unstructured
\end{table}

%Describe the results
The model selection performance metrics of the penalized G-estimator are provided in Table~\ref{tab.selection} in Appendix~\ref{A.add.sim.res}. We report the average CI lengths and the false coverage rates in Figure~\ref{fig.averageCI} and Figure~\ref{fig.fcr}, respectively,   and the power in Table~\ref{tab.power}. These metrics were obtained using different inferential methods under three distinct correlation structures for various simulation settings. The naive approach produced the smallest confidence intervals, followed closely by the one-step estimators. Although we see 100\% power for selecting non-zero blip coefficients under the naive approach, the false coverage rates under this approach exceeded the nominal significance level (0.05), especially in small samples and when the number of variables was large. Confidence intervals produced by the UPoSI approach were far wider than those of the other competing methods. While the UPoSI method yielded false coverage rates lower than 0.05 in all of the simulations settings, its power to detect true effects was nearly zero in small samples. The power under UPoSI increased with increasing sample size with a low dimensional number of candidate effect modifiers, but when we had high-dimensional covariates the power was far smaller in comparison to the other methods. Inference based on the one-step improved penalized G-estimator, given that the sparse weight vector is estimated by either LASSO or Dantzig selector, provided false coverage rates lower than 0.05, and also resulted in strong power for the non-zero coefficients selected by the initial penalized method.

We also investigated the performance of the inferential methods under a simulation setting that triggers the G-null paradox. The G-null paradox highlights the possibility that even under the global null of no treatment effect, biased nonzero estimates may arise when both nuisance models are misspecified. The results from our investigation on the G-null paradox is presented in Appendix~\ref{A.gNULL}.

\section{Application} \label{sec.application}
In this section, we illustrate our methodological developments using the cohort and data previously analyzed by \cite{jaman2025penalized}. The cohort consists of patients undergoing chronic hemodiafiltration at the CHUM and the CED, %. We consider Hemodiafiltration as chronic if the patient has at least 28 consecutive sessions. The cohort start-date for each patient was their first dialysis session on or after  March 1st, 2017; the cohort end date was December 1st, 2021. Approximately 215 patients were treated at the outpatient centre (CED) with three and, for some, four dialysis sessions per week. The primary data include information from a total of 474 patients who underwent a total of 170761  dialysis sessions. 
who started their treatment on or after March 1, 2017, and were followed through December 1, 2021, with a total of 474 patients and 170,761 dialysis sessions recorded. \cite{jaman2025penalized} explored the impact of dialysis facility (CHUM vs.\ CED) on session-specific mean convection volumes, with a focus on effect modification by patient characteristics, and found the possibility of effect modification by cancer status. %In this paper, we aim to provide valid inference on the effect of dialysis facility and on the nature and magnitude of effect modifications previously identified. 
This paper extends the prior analysis by providing valid inference and ensuring robust conclusions about the impact of dialysis facility and effect heterogeneity by patient-level factors. %The information extracted from hospital databases for each session includes: drugs (BDM), laboratory results (CERNER), procedures related to dialysis venous access (Radimage), and dialysis-related variables and dialysis-related drugs (EuCliD-NephroCare). We received the cleaned study dataset comprised of time-varying elements of the hemodiafiltration prescription predicting convection volume, concomitant therapies for managing anemia or low albumin or maintaining dialysis access patency, and dialysis session-specific results. Comorbidities that are potential confounders were obtained from the \emph{Maintenance et exploitation des données pour l'étude de la clientèle hospitalière} (MED-ECHO) database. Confounders and potential effect modifiers that we considered in the analysis are previous outcome (24L or less$= 1$ vs.\ more than 24L $= 0$), hemoglobin, albumin, dalteparin, access (fistula = 0, catheter = 1), catheter change, age, sex, and the components of the Charlson Comorbidity Index (hypertension, diabetes, peripheral vascular disease, congestive heart failure, cardiac arrhythmia, acute myocardial infarction, chronic pulmonary disease, liver disease, valvular disease, cancer, metastatic cancer, cerebrovascular disease, dementia, %Peptic ulcer disease, %hemiplegia, and rheumatic disease). The exposure is the dialysis facility, coded as one if the treatment location was CHUM and zero if it was CED. %Some descriptive statistics of the data are presented in Section S9 of the Supplementary Materials.

The data extracted from hospital databases for each session includes time-dependent variables for hemodiafiltration prescriptions and dialysis session-specific outcomes, along with detailed patient characteristics, such as hemoglobin levels, albumin, comorbidities according to the Charlson Index (cancer status, hypertension, diabetes, etc.), and dialysis access type (fistula vs. catheter). For a detailed description of the candidate effect modifiers and additional information about the penalized estimation, we refer readers to the application section in \cite{jaman2025penalized}.
%Our aim was to estimate proximal effects on hemodiafiltration outcomes, and because we did not expect substantial intra-patient variability, we did not expect to gain substantially from having a large number of sessions per patient. We were also concerned with the computer memory and computational issues that would arise from incorporating all of the numerous sessions per patient (which could be as many as 827). Hence, we included the first six consecutive sessions of post-dilution hemodiafiltration for each of the patients in our analysis. We estimated the propensity scores using a logistic regression of exposure conditional on all of the potential confounders using the pooled data set. We performed the proposed penalized G-estimation considering four different correlation structures: independent, exchangeable, autoregressive of order one (AR1), and unstructured.
For each proximal effect estimate, we obtained confidence intervals using the UPoSI approach, the decorrelated score method, and the naive sandwich variance estimator. Consistent with \cite{jaman2025penalized}, we performed our analysis employing the same four correlation structures: independent, exchangeable, autoregressive of order one (AR1), and unstructured. We set $\text{Corr}(U_{ij},U_{ij'}) = \rho^{|j-j'|}$ for $j\neq j'$ under the AR1 structure. The specifications for the other three structures are already provided in Section~\ref{sec.simulation} (simulation study).

The selected blip model under the AR1 correlation structure with adjustment for all potential confounders in the treatment-free part is:
\begin{align*}
    \gamma_j(a_{j}, \vec h_{j};\vec\psi)
    %\text{blip}
    = (\psi_0 + \psi_1\times\text{Cancer}_j)\,\text{CHUM}_j
\end{align*}
where $j=1,2,\ldots,6$. 
The estimates of the blip parameters are given in Figure~\ref{fig.hdf} with associated 95\% confidence intervals obtained using the candidate inferential methods.
\begin{figure}[t]
    \centering
    \includegraphics[scale=0.8]{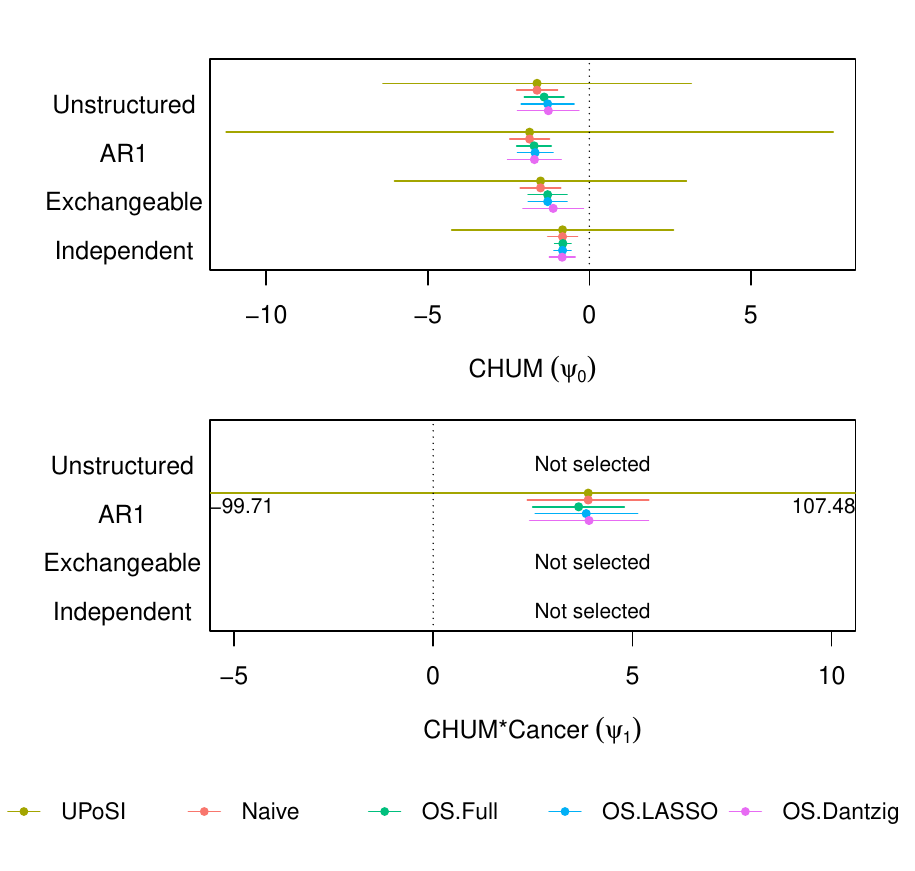}
    \caption{Estimated blip parameters and their corresponding 95\% confidence intervals obtained using the competing inferential methods under different working correlation structures for the hemodiafiltration study.}
    \label{fig.hdf}
\end{figure}
The estimated main effect of the dialysis facility was negative across all the working correlation structures considered, with the associated 95\% confidence intervals excluding zero for all competing methods except the UPoSI. Cancer was identified as an effect modifier only under the AR1 correlation structure, with the corresponding 95\% confidence interval again excluding zero for all methods except the UPoSI. Results under AR1 structure suggest that the effect of dialysis facility on the convection volume differs by the cancer status of the patient. For patients without cancer, the mean convection volume was 1.85 litres lower at CHUM compared to CED, after adjusting for all other confounders. However, for cancer patients, the mean convection volume was $3.89-1.85=2.04$ litres higher at the CHUM.

\section{Discussion} \label{sec.discussion}

We extended and evaluated two post-selection inferential methods for valid inference on the effect modification of proximal treatment effects estimated through penalized G-estimation. The one-step improved penalized G-estimator with a sparse weight vector showed good performance, providing valid inferential guarantees for the target parameters. Using the competing inferential methods, we investigated if the effect of dialysis facility on dialysis outcome (convection volume) differed by the demographics, clinical characteristics, and comorbidity status of patients with end-stage renal disease. Our findings suggest that while the CED generally achieved better hemodiafiltration outcomes, cancer patients with similar measured characteristics might have better outcomes at CHUM compared to CED.

The poor power of random design UPoSI in our simulation results can be well understood from the construction of the confidence region under this method. The confidence region or the coordinate-wise confidence interval according to UPoSI involves the $L_1$ norm of the full parameter vector, which increases with an increasing number of covariates and leads to excessively wide confidence intervals. Further theoretical work is needed to enhance the power of the UPoSI method. Although fixed-design UPoSI produces less conservative confidence intervals as the uncertainty component related to the covariates becomes zero, the dynamic nature of the treatment restricted us from considering the fixed-design UPoSI in our context. In high-dimensional settings, restricting the candidate set to plausible models can mitigate excessive conservatism of this method. Our simulation results demonstrated that the one-step improved penalized G-estimator with sparse weights estimated via the LASSO effectively controls false coverage rates. However, these results may not hold under other forms of misspecification in the outcome model, such as missing exponential terms of confounders. We recommend the Dantzig selector for estimating the sparse weights, as it provides higher-order corrections. It is important to note that the oracle properties of the penalized G-estimator rely on minimal signal strength conditions. Conditional methods like UPoSI, or the naive sandwich estimator can not quantify uncertainty for weak signals missed in the regularized estimation. In contrast, the decorrelated score method does not require variable selection consistency and provides reliable uncertainty estimates for small signals.

Under the decorrelated score approach, constructing separate confidence intervals for each $\psi_k$ may raise concerns about multiple comparisons, particularly when the parameters are interpreted jointly. While each interval is typically constructed to achieve, for example, 95\% coverage individually, the probability that at least one interval fails to cover the true parameter value increases as the number of parameters grows. This is the well-known multiple comparisons problem. Interpreting multiple $\psi_k$ simultaneously, or making claims about which effects are statistically significant, can therefore lead to inflated Type I error rates. This issue is less concerning when each $\psi_k$ corresponds to a pre-specified and scientifically distinct hypothesis, and no joint inference is intended. When necessary, multiplicity adjustments, such as the Bonferroni or Holm procedures, or false discovery rate (FDR) control using the Benjamini-Hochberg method, can be applied to address this issue. In contrast, UPOSI provides valid inference for the full model simultaneously (as formalized in Theorem~\ref{thm.asymp.validity.UPoSI}), enabling coherent joint inference without requiring post hoc multiplicity adjustments.

Future research may explore the robustness of the proposed methods under broader forms of misspecification in the outcome model. Extending these inferential methods to causal frameworks beyond effect modification analysis, including approaches like instrumental variable analyses and mediation analyses, would be a potential future direction.

\section*{Software implementation}
The R-packages for implementing our proposed methods, \texttt{UPoSIPeG} and \texttt{OSIPeG}, are available on GitHub at \url{https://github.com/ajmeryjaman/UPoSIPeG/} and \url{https://github.com/ajmeryjaman/OSIPeG/}, respectively. Both packages contain implementation code as well as illustrative examples that demonstrate how to apply the methods to data.

\section*{Funding}
This work is supported by a doctoral scholarship from the Fonds de Recherche du Québec Nature et technologies (FRQNT) of Canada to AJ and a Discovery Grant from the Natural Sciences and Engineering Research Council of Canada to MES. AE is supported by the research grants (R01DA058996, R01DA048764, and R33NS120240) from the National Institutes of Health. MES is supported by a tier 2 Canada Research Chair.

%\cfoot{Bibliography}
%\rfoot{1}
\bibliographystyle{apalike}
\bibliography{bibliography.bib}

\clearpage

\appendix
\section{Appendix: Additional technical details}\label{A.add.tech.detail}
\setcounter{table}{0}
\setcounter{theorem}{0}
\setcounter{equation}{0}
\renewcommand{\thetable}{\thesection\arabic{table}}
\renewcommand\thetheorem{\thesection.\arabic{theorem}}
\renewcommand{\theassumption}{\thesection.\arabic{assumption}}
\renewcommand{\theequation}{\thesection.\arabic{equation}}

\subsection{Assumptions for identifiability of the SNMM parameters}\label{A.assump}
Assumptions \ref{asmp1}-\ref{asmp3} are the usual causal assumptions for identifiability of the target parameters using the observed data.

\begin{assumption}[Consistency]\label{asmp1}
 The observed outcome is equal to the potential outcome at occasion $j$, for $j=1,\dots,J$, if the observed treatment history matches the counterfactual history at occasion $j$, i.e., $Y_{ij}(\overline{a}_j) = Y_{ij}$, if $\overline{A}_{ij}=\overline{a}_{j}$.   
\end{assumption}

\begin{assumption}[Sequential ignorability]\label{asmp2}
 The potential outcome $Y_{ij}(\overline{a}_{j-1},0)$ is independent of $A_{ij}$ conditional on $\vec H_{ij}$ and $\Bar{A}_{i,j-1}$, for $j=1,\dots,J$.   
\end{assumption}

\begin{assumption}[Positivity]\label{asmp3}
  If the joint density of $\vec H_{ij}$ at $\{\vec h_{ij}\}$ is greater than zero, then $\mathbb{P}(A_{ij}=a_j|\vec H_{ij}=\vec h_{ij}) > 0$ for all $a_j$, $j=1,\dots,J$.  
\end{assumption}

Note that, the SNMM models the causal effect as the expected difference between treated and untreated potential outcomes as shown in (\ref{snmm.proximal}) in the manuscript. Rearranging the terms in this equation we obtain the following:
\begin{align*}
        E&[Y_{ij}(\Bar{a}_{j-1}, 0)|\vec H_{ij}=\vec h_{ij},\Bar{A}_{ij}=\Bar{a}_{j}]\\
        &= E[Y_{ij}(\Bar{a}_{j-1}, a_j)-\gamma_j^*(a_{j},\vec h_{ij};\vec\psi)|\vec H_{ij}=\vec h_{ij},\Bar{A}_{ij}=\Bar{a}_{j}]\\
        &=E[Y_{ij}-\gamma_j^*(a_{j},\vec h_{ij};\vec\psi)|\vec H_{ij}=\vec h_{ij},\Bar{A}_{ij}=\Bar{a}_{j}];\qquad[\text{by consistency assumption}]\\
        &=E[U_{ij}|\vec H_{ij}=\vec h_{ij},\Bar{A}_{ij}=\Bar{a}_{j}];\qquad[\text{by definition}].
        \end{align*}
Under G-estimation, the estimating equations are constructed based on the blipped-down outcomes $U_{ij}$. If $Y_{ij}(\overline{a}_{j-1},0) \perp A_{ij}|\vec H_{ij},\Bar{A}_{i,j-1}$ or equivalently $U_{ij} \perp A_{ij}|\vec H_{ij},\Bar{A}_{i,j-1}$, this independence gives the necessary moment condition for G-estimation so that we have $E[\vec S^{\text{eff}}(\vec\psi)]=\vec 0$. The blipped-down outcome $U_{ij}$ removes the causal effect of $A_{ij}$. After removing the effect, the residual (the counterfactual under no treatment) is orthogonal to the treatment assignment given history. This orthogonality identifies $\vec\psi$ via G-estimation.

\subsection{Technicalities related to the working correlation matrix}\label{A.sec.est.R}

For subject $i$ the correlation matrix $\vec R_i(\rho)$ is unknown and is replaced by the estimate $\widehat{\vec R}$ while performing penalized G-estimation. We need the following assumption on the correlation matrix for asymptotic validity of the proposed inferential methods.

\begin{assumption}\label{asmp4}
 The common true correlation matrix $\vec R_0$ for the observed outcomes has eigen values bounded away from zero and $+\infty$. The estimated working correlation matrix $\hat{\vec R}$ satisfies $||\hat{\vec R}^{-1} - \overline{\vec R}^{-1}|| = O_{\mathbb{P}}(\sqrt{1/n})$, 
    %$\hat{\vec R} - \overline{\vec R} \overset{L^1}{\rightarrow} 0$ (elementwise), 
where $\overline{\vec R}$ is a constant positive definite matrix with eigen values bounded away from zero and $+\infty$, and $||\cdot||$ denotes the Frobenius norm.%Here, $s$ represents the number of unpenalized main effects plus the number of penalized modified effects which are truly nonzero.   
\end{assumption}

Under Assumption~\ref{asmp4}, the two versions of the expected information have the following forms:  
\begin{align*}
\vec H(\vec\theta)&= \text{E}\{-\partial\vec S_i^{\text{eff}}(\vec\theta)/\partial\vec\theta^\top\} = \vec D_{i}^\top\vec Q_i^{-1/2}\overline{\vec R}^{-1}\vec Q_i^{-1/2}(\vec H_i\;\;\vec A_i\circledast\vec H_i),\\
\vec I(\vec\theta)&= \text{E}\{\vec S_i^{\text{eff}}(\vec\theta)\vec S_i^{\text{eff}}(\vec\theta)^\top\} = \vec D_{i}^\top\vec Q_i^{-1/2}\overline{\vec R}^{-1}\vec R_0\overline{\vec R}^{-1}\vec Q_i^{-1/2}\vec D_{i}.
\end{align*}
Details regarding the expressions of these information matrices can be found in \cite{balan2005asymptotic}, where the authors presented a rigorous asymptotic theory for generalized estimating equations.

%\subsection{Regularity conditions} \label{regularity.cond}
Some other regularity conditions for the desired asymptotic properties of the penalized G-estimator are as follows:
\begin{enumerate}
    \item [(C1)] All variables in $\vec D_{ij}$, $i=1,\ldots,n$, $j=1,\ldots,J$, are uniformly bounded.
    \item [(C2)] The unknown parameter $\vec\theta_n$ belongs to a compact subset $\vec\Theta \subseteq R^{2K}$ and the true parameter $\vec\theta_0$ lies in the interior of $\vec\Theta$.
    \item [(C3)] There exists finite positive constants $c_1$ and $c_2$ such that
    \begin{align*}
        c_1 \leq \omega_{min} \Bigg(\dfrac{\sum_{i=1}^n\vec D_i^\top(\vec H_i\;\;\vec A_i\circledast\vec H_i)}{n}\Bigg) \leq \omega_{max} \Bigg(\dfrac{\sum_{i=1}^n\vec D_i^\top(\vec H_i\;\;\vec A_i\circledast\vec H_i)}{n}\Bigg)\leq c_2,
    \end{align*}
    where $\omega_{min}(\vec D)$ and $\omega_{max}(\vec D)$ denote the minimum and maximum of the eigenvalues, respectively, of the matrix $\vec D$.
    \item [(C4)] Let $\vec\xi_i(\vec\theta_n) = (\vec\xi_{i1}(\vec\theta_n),\ldots,\vec\xi_{in_i}(\vec\theta_n))^\top=\vec Q_i^{-1/2}(\vec Y_i - \vec g_i(\vec\theta_n))$. There exists a finite constant $d_1 > 0$ such that $E(||\vec\xi_i(\vec\theta_0)||^{2+\rho}) \leq d_1$ for all $i$ and some $\rho > 0$; and there exists positive constants $d_2$ and $d_3$ such that $E(\exp(d_2|\xi_{ij}(\vec\theta_0))|\vec D_i) \leq d_3$, uniformly in $i=1,\ldots,n$, $j=1,\ldots,J$. 
    \item [(C5)] Let $T_n=\{\vec\theta_n: ||\vec\theta_n - \vec\theta_0||\leq \Delta\sqrt{1/n}\}$, then $g'(\vec D_{ij}\vec\theta_n)$, $i=1,\ldots,n$, $j=1,\ldots,J$, are uniformly bounded away from 0 and $\infty$ on $T_n$; $g''(\vec D_{ij}\vec\theta_n)$ and $g'''(\vec D_{ij}\vec\theta_n)$, $i=1,\ldots,n$, $j=1,\ldots,J$, are uniformly bounded by a finite positive constant $d_2$ on $T_n$; $g'(.)$, $g''(.)$ and $g'''(.)$ denote the first, second and third derivatives of the function $g(.)$, respectively.
    \item [(C6)] Suppose $s$ denotes the number of the non-zero blip coefficients plus the number of other fixed parameters in the outcome mean model. When $s$ is not fixed, assuming $\min_{m \in B} |\theta_{0m}|/\lambda_n \rightarrow \infty$ as $n \rightarrow \infty$ and $s_n^3n^{-1} = o(1)$, $\lambda_n \rightarrow 0$, $s_n^2(\log n)^4 = o(n\lambda_n^2)$, $\log(K_n) = o(n\lambda_n^2/(\log n)^2)$, $K_ns_n^4(\log n)^6 = o(n^2\lambda_n^2)$, and $K_ns_n^3(\log n)^8 = o(n^2\lambda_n^4)$. Note that $\lambda_n$ is the tuning parameter.
\end{enumerate}
These conditions are similar to the regularity conditions in \cite{wang2012penalized,jaman2025penalized}, some of which maybe further relaxed.

\subsection{Technical details related to the UPoSI approach}\label{A.tech.detail.uposi}
%% State required inequalities to be used later

We follow \cite{kuchibhotla2020valid} to state and prove asymptotic validity of the UPoSI method. We will require the following inequality:
\begin{align} \label{req.ineq}
    ||B v||_\infty \leq ||B||_\infty ||v||_1.
\end{align}

%% State and prove the theorem of equivalance (THM 3.1)
\begin{theorem} \label{thm.equiv}
    For a given sets of models $\mathcal{M}_K$, any set of confidence regions $\{\widehat{\mathcal{R}}_{n,M} : M \in \mathcal{M}_K\}$, and significance level $\alpha \in [0,\, 1]$, the statements (a) and (b) are equivalent:
    \begin{enumerate}
        \item[(a)] The post-selection inference problem is solved, meaning that
\begin{align*}
    \mathbb{P}(\vec\theta_{n,\widehat{M}} \in \widehat{\mathcal{R}}_{n,\widehat{M}}) \geq 1 - \alpha.
\end{align*}
        \item[(b)] The simultaneous inference problem for $M \in \mathcal{M}_K$ is solved, meaning that
\begin{align*}
    P\Bigg(\bigcap_{M \in \mathcal{M}_K}\{\vec\theta_{n,M} \in \widehat{\mathcal{R}}_{n,M}\}\Bigg) \geq 1 - \alpha.
\end{align*}
    \end{enumerate}
\end{theorem}

\begin{proof}
Let $\mathcal{F}_M=\{\vec\theta_{n,M}\in \widehat{\mathcal{R}}_{n,M}\}$ denote one coverage event inside (b) for a fixed $M \in \mathcal{M}_K$. For a random model $\widehat M$, let $\mathcal{F}_{\widehat M}=\{\widehat\theta_{n,\widehat M}\in \widehat{\mathcal{R}}_{n,\widehat M}\}$ be the coverage event inside (a). Both are random events since both of the confidence regions are random.

Note that $\mathcal{F}_{\widehat M} \supseteq \bigcap_{M \in \mathcal{M}_K} \mathcal{F}_M$ since $\widehat M \in \mathcal{M}_K$. Hence, if (b) is true, it implies that (a) is also true.

To prove the converse, it is sufficient to construct a data-driven selection procedure $\widehat M$ that satisfies
\begin{align} \label{suff.cond.equiv}
    \mathcal{F}_{\widehat M} = \bigcap_{M \in \mathcal{M}_K} \mathcal{F}_M.
\end{align}
Let $\widehat M$ be any selection procedure that satisfies
\begin{align*}
    \widehat M \in \underset{M \in \mathcal{M}_K}{\text{arg min }} \bm{1}\{\mathcal{F}_{M}\},
\end{align*}
where $\bm{1}\{E\}$ represents the indicator function of the event $E$. It follows that $\bm{1}\{\mathcal{F}_{\widehat M}\}=\text{min}_{M \in \mathcal{M}_K} \bm{1}\{\mathcal{F}_{M}\}$, which is equivalent to (\ref{suff.cond.equiv}). Hence, (a) implies (b).
\end{proof}

\begin{theorem} \label{thm.UPoSI.region.derivation}
The UPoSI confidence regions $\{\widehat{\mathcal{R}}_{n,M} : M \in \mathcal{M}_K(k)\}$ defined in (9) in the manuscript, satisfy 
\begin{align} \label{conf.statement.simultn}
    P\Bigg(\bigcap_{M \in \mathcal{M}_K(k)}\{\vec\theta_{n,M} \in \widehat{\mathcal{R}}_{n,M}\}\Bigg) \geq 1 - \alpha.
\end{align}
Also, for any selected model $\widehat{M}$, where $\widehat{M} \in \mathcal{M}_K(k)$, the following is satisfied
\begin{align}\label{conf.statement.select}
    \mathbb{P}(\vec\theta_{n,\widehat{M}} \in \widehat{\mathcal{R}}_{n,\widehat{M}}) \geq 1 - \alpha.
\end{align}
\end{theorem}

\begin{proof}
This proof is free of stochastic assumptions. If we subtract the expected equation (7) from the empirical equation (6), for any $M \in \mathcal{M}_K(k)$ the following holds:
\begin{align*}
    \widehat{\vec W}_n(M)\widehat{\vec\theta}_{n,M} - \vec W_n(M)\vec\theta_{n,M} = \widehat{\vec G}_n(M) - \vec G_n(M).
\end{align*}
If we subtract and add $\widehat{\vec W}_n(M)\vec\theta_{n,M}$ on the left-hand side of this equation we have
\begin{align*}
    \widehat{\vec W}_n(M)(\widehat{\vec\theta}_{n,M} - \vec\theta_{n,M}) 
    + \big(\widehat{\vec W}_n(M) - \vec W_n(M)\big)\vec\theta_{n,M} = \widehat{\vec G}_n(M) - \vec G_n(M).
\end{align*}
If we move the second term from the left to the right-hand side of the equality, take the sup norm and apply the triangle inequality on the right-hand side, we get the following:
\begin{align*}
    ||\widehat{\vec W}_n(M)(\widehat{\vec\theta}_{n,M} - \vec\theta_{n,M})||_\infty \leq ||\widehat{\vec G}_n(M) - \vec G_n(M)||_\infty + ||\big(\widehat{\vec W}_n(M) - \vec W_n(M)\big)\vec\theta_{n,M}||_\infty. 
\end{align*}
If we use the inequality (\ref{req.ineq}) to the last term it follows that
\begin{align*}
    ||\widehat{\vec W}_n(M)(\widehat{\vec\theta}_{n,M} - \vec\theta_{n,M})||_\infty \leq ||\widehat{\vec G}_n(M) - \vec G_n(M)||_\infty + ||\widehat{\vec W}_n(M) - \vec W_n(M)||_\infty ||\vec\theta_{n,M}||_1. 
\end{align*}
Since $\widehat{\vec W}_n(M) - \vec W_n(M)$ and $\widehat{\vec G}_n(M) - \vec G_n(M)$ are a submatrix and a subvector of $\widehat{\vec W}_n - \vec W_n$ and $\widehat{\vec G}_n - \vec G_n$, respectively, we can write
\begin{align}\label{events.ineq}
    ||\widehat{\vec W}_n(M)(\widehat{\vec\theta}_{n,M} - \vec\theta_{n,M})||_\infty \leq ||\widehat{\vec G}_n - \vec G_n||_\infty + ||\widehat{\vec W}_n - \vec W_n||_\infty ||\vec\theta_{n,M}||_1. 
\end{align}
This inequality is true for any sample and for any submodel $M \in \mathcal{M}_K(k)$. These enable us to take the intersection of the events (\ref{events.ineq}) over all possible submodels and transform it into a ``probability one" statement. Using the definitions of $D_n^{G}$ and $D_n^{W}$, we have
\begin{align}
    P\Bigg(\bigcap_{M \in \mathcal{M}_K(k)}\bigg\{||\widehat{\vec W}_n(M)(\widehat{\vec\theta}_{n,M} - \vec\theta_{n,M})||_\infty \leq  D_n^{G} +  D_n^{W}||\vec\theta_{n,M}||_1\bigg\}\Bigg) = 1.
\end{align}
Considering the definitions of $C_n^{G}(\alpha)$ and $C_n^{W}(\alpha)$ the proof of (\ref{conf.statement.simultn}) is complete. The proof of (\ref{conf.statement.select}) follows by an application of Theorem~\ref{thm.equiv}.
\end{proof}

%% Asymptotic validity

\begin{proof}[\bf Proof of Theorem~\ref{thm.asymp.validity.UPoSI}]
For all $k \geq 1$ satisfying $kD_n^{W} \leq \omega_n(k)$ and for all $M \in \mathcal{M}_K(k)$, if $\widehat{\vec\theta}_{n,M}$ is an uniform-in-model consistent estimator of $\vec\theta_{n,M}$ then
    \begin{align} \label{uniform.cond}
        ||\widehat{\vec\theta}_{n,M} - \vec\theta_{n,M}||_1 \leq \frac{|M|(D_n^{G} +  D_n^{W}||\vec\theta_{n,M}||_1)}{\omega_n(k)-kD_n^{W}}.
    \end{align}
Under the assumption on the minimum eigen value equation (\ref{uniform.cond}) implies that for all $M \in \mathcal{M}_K(k)$,
\begin{align*}
    \Bigg|\frac{D_n^{G} +  D_n^{W}||\widehat{\vec\theta}_{n,M}||_1}{D_n^{G} +  D_n^{W}||\vec\theta_{n,M}||_1} - 1\Bigg| 
    &\leq \frac{D_n^{W}||\widehat{\vec\theta}_{n,M} - \vec\theta_{n,M}||_1}{D_n^{G} +  D_n^{W}||\vec\theta_{n,M}||_1}\\
    &\leq \frac{D_n^{W}}{D_n^{G} +  D_n^{W}||\vec\theta_{n,M}||_1} \cdot \frac{|M|(D_n^{G} +  D_n^{W}||\vec\theta_{n,M}||_1)}{\omega_n(k)-|M|D_n^{W}}\\
    &\leq \frac{kD_n^{W}}{\omega_n(k)-kD_n^{W}}.
\end{align*}
Therefore,
\begin{align*}
   \underset{M \in \mathcal{M}_K(k)}{\text{sup}} \Bigg|\frac{D_n^{G} +  D_n^{W}||\widehat{\vec\theta}_{n,M}||_1}{D_n^{G} +  D_n^{W}||\vec\theta_{n,M}||_1} - 1\Bigg| \leq \frac{kD_n^{W}/\omega_n(k)}{1-(kD_n^{W}/\omega_n(k))} = o_{\mathbb{P}}(1).
\end{align*}
Hence,
\begin{align*}
    \underset{n \rightarrow \infty}{\text{lim inf }}     P\Bigg(\bigcap_{M \in \mathcal{M}_K(k)}\Big\{
    ||\widehat{\vec W}_n(M)(\widehat{\vec\theta}_{n,M} - \vec\theta_{n,M})||_\infty \leq  D_n^{G} +  D_n^{W}||\widehat{\vec\theta}_{n,M}||_1\Big\}\Bigg) = 1.
\end{align*}
\end{proof}
Following the definitions of $C_n^{G}(\alpha)$ and $C_n^{W}(\alpha)$ we have the required result.

%% Need to mention the assumptions on the multiplier weights

\subsubsection{Multiplier bootstrap} \label{A.sec.multp.boot}

The computation of the UPoSI confidence regions~(\ref{uposi.region.asymp}) depends on the estimation of the joint quantiles $C_n^{G}(\alpha)$ and $C_n^{W}(\alpha)$ using the data.  The multiplier bootstrap is a fast and easy-to-implement alternative to the standard resampling bootstrap. The use of multiplier bootstrap to estimate these quantiles can be justified by an application of the high-dimensional central limit theorem \citep{kuchibhotla2020valid}. The applicability of multiplier bootstrap for estimating the standard error of parameter estimates when analyzing clustered data using GEE has been mentioned in \cite{li2008smooth} and \cite{cheng2013cluster}. 

We define the subject-specific vector $\vec Z_i$ that contains the contribution of subject $i$ to $\widehat{\vec W}_n$ and $\widehat{\vec G}_n$ as follows:
\begin{align*}
  \vec Z_i =
  &\Bigg(\Big\{\vec h_{i[k]}^\top\vec V_i^{-1}\vec y_i\Big\}_{1\leq k\leq K}\text{ , }
  \Big\{\Big[\{\vec a_i-E(\vec A_i|\vec H_i)\}\odot\vec h_{i[k]}\Big]^\top\vec V_i^{-1}\vec y_i\Big\}_{1\leq k\leq K}\text{ , }
  \Big\{\vec h_{i[k]}^\top\vec V_i^{-1}\vec h_{i[k']}\Big\}_{1\leq k \leq k' \leq K}\text{ , }\\
  &\Big\{\vec h_{i[k]}^\top\vec V_i^{-1}(\vec a_i\odot\vec h_{i[k']})\Big\}_{1\leq k \leq k' \leq K}\text{ , }
  \Big\{\Big[\{\vec a_i-E(\vec A_i|\vec H_i)\}\odot\vec h_{i[k]}\Big]^\top\vec V_i^{-1}\vec h_{i[k']}\Big\}_{1\leq k \leq k' \leq K}\text{ , }\\
  &\Big\{\Big[\{\vec a_i-E(\vec A_i|\vec H_i)\}\odot\vec h_{i[k]}\Big]^\top\vec V_i^{-1}(\vec a_i\odot\vec h_{i[k']})\Big\}_{1\leq k \leq k' \leq K}
  \bigg),
\end{align*}
where $\vec h_{i[k]}$ represents the $k$-th column vector of the matrix $\vec h_i$ for $k=1,\ldots,K$, and ``$\odot$" denotes the element-wise product. The number of elements in $\vec Z_i$ is $K+K+4\times\{K+K(K-1)/2\}=O(K^2)$. We define an event $\{D_n^G \leq d_1 \text{ and }D_n^W \leq d_2\}$ for constructing the the bivariate quantiles for $D_n^G$ and $D_n^W$. As shown by \cite{kuchibhotla2020valid}, this event for any $d_1,d_2 \geq 0$ can be written as a symmetric rectangle in terms of
\begin{align*}
    S_n^Z=\frac{1}{\sqrt{n}}\sum_{i=1}^n\{\vec Z_i-E(\vec Z_i)\}.
\end{align*}
Let $r_1, r_2, \ldots, r_n$ be independent standard normal random variables and define
\begin{align*}
    S_n^{r,Z} = \frac{1}{\sqrt{n}}\sum_{i=1}^n r_i(\vec Z_i - \overline{\vec Z}_n) \text{ with } \overline{\vec Z}_n = \frac{1}{n}\sum_{i=1}^n \vec Z_i.
\end{align*}
Also let $S_n^{r,Z}(I)$ represents the first $2K$ elements of $S_n^{r,Z}$ that contribute towards estimation of $\vec G_n$ and $S_n^{r,Z}(II)$ represents the remaining elements of $S_n^{r,Z}$ that contribute towards estimation of $\vec W_n$. Then the joint quantiles can be estimated using the following steps:
\begin{enumerate}
    \item Generate $R_n$ random vectors of dimension $n$ from a standard normal distribution and denote the generations with $r_{i,j}$ for $i=1,2,\ldots,n$ and $j=1,2,\ldots,R_n$.
    \item Compute the $j$-th replicate of $S_n^{r,Z}$ as
    \begin{align*}
        S_{n,j}^* = \frac{1}{\sqrt{n}}\sum_{i=1}^n r_{i,j}(\vec Z_i - \overline{\vec Z}_n) \text{ for } j=1,2,\ldots, R_n.
    \end{align*}
    \item Find any two quantities $\widehat{C}_n^G(\alpha)$ and $\widehat{C}_n^W(\alpha)$ such that
    \begin{align*}
        \frac{1}{R_n}\sum_{j=1}^{R_n} \bm{1}\Big\{||S_{n,j}^*(I)||_{\infty} \leq \widehat{C}_n^G(\alpha), ||S_{n,j}^*(II)||_{\infty} \leq \widehat{C}_n^W(\alpha)\Big\} \geq 1-\alpha,
    \end{align*}
    where $\bm{1}\{E\}$ represents the indicator function of the event $E$.
\end{enumerate}
\noindent
\cite{cheng2013cluster} provided theoretical proof of the estimation consistency for the exchangeably weighted cluster bootstrap method for GEE. Multiplier bootstrap can be viewed as a special class of the weighted bootstrap and it satisfies all the required conditions on the weights for consistency to hold \citep{cheng2013cluster}. However, we only expect the asymptotic conservativeness of the proposed multiplier bootstrap instead of consistency, because we replace $E(\vec Z_i)$ with $\Bar{\vec Z}_n$ which is not a consistent estimator.

\subsubsection{Coordinate-wise confidence interval under UPoSI} \label{A.sec.conf.int}
We can construct  coordinate-wise confidence interval \citep{kuchibhotla2020valid} for the $k$-th coefficient in $\vec\psi$ having the form $\widehat{\psi}_k \pm \widehat{\mathcal{L}}_{k,\widehat{M}}$ under any selected model $\widehat{M}$ (could be any submodel $M \in \mathcal{M}$), where $\widehat{\psi}_k$ represents the $k$-th element of the target estimate $\widehat{\vec\psi}$ in the penalized G-estimate $\widehat{\vec\theta} = (\widehat{\vec\delta},\widehat{\vec\psi})^\top$ and $\widehat{\mathcal{L}}_{k,\widehat{M}}$ represents the half-length of the confidence interval. We can compute the half-length of the interval as follows:
\begin{align*}
  \widehat{\mathcal{L}}_{k,\widehat{M}} = \Bigg|\vec c_k'\Big\{\widehat{\vec W}_n(\widehat{M})\Big\}^{-1}\Bigg| \Bigg(\widehat{C}_n^G(\alpha) + \widehat{C}_n^W(\alpha)||\widehat{\vec\theta}||_1\Bigg) 
\end{align*}
where $\vec c_k$ is a vector of zeros with value 1 corresponding to the position $k$.

\subsection{Technical details related to the decorrelated score method} \label{A.tech.detail.one.step}
%The assumptions, the theorem of asymptotic normality and the proof
Let $\vec\theta^* = (\vec\delta^{*\top}, \vec\psi^{*\top})^\top$ denote the true values of $\vec\theta = (\vec\delta^{\top}, \vec\psi^{\top})^\top$, $\vec S_{\vec\psi^*} = \vec S_{\vec \psi}(\vec \theta^*)$ be the sub-vector of $S(\vec\theta^*)$ corresponding to $\vec\psi$, and $\vec I^* = E[\vec S_{\vec\psi^*}\vec S_{\vec\psi^*}^\top]$. Recall that for inference regarding $\psi_k$, i.e., the parameter of interest, we made a partition of the target parameter vector as $\vec\psi=(\psi_k,\vec\nu_k)$, where  $k$ can take any value in $\{0,1,\ldots, K-1\}$ and $\vec\nu_k = (\psi_0,\ldots,\psi_{k-1},\psi_{k+1},\ldots, \psi_{K-1})$. We also define $\vec w^*=\vec I_{\vec\nu_k\vec\nu_k}^{*-1}\vec I_{\psi_k\vec\nu_k}^*$. First we state the assumptions required for the validity of the target inference. Assumptions \ref{asmp5}-\ref{asmp8} are similar to the Assumptions 3.1-3.4 in \cite{ning2017general} required to establish the asymptotic normality of the one-step improved estimator.

\begin{assumption}[Consistency conditions for initial penalized G-estimator]\label{asmp5}
 For some sequences $\eta_1(n)$ and $\eta_2(n)$ converging to 0 as $n \rightarrow \infty$ the following holds
\begin{align*}
\underset{n \rightarrow \infty}{\text{lim}} \mathbb{P}_{\vec\psi^*}\big(\big|\big|\widehat{\vec\psi} - \vec\psi^*\big|\big|_1 \lesssim \eta_1(n)\big) =1 \text{\hspace{0.5cm}and\hspace{0.5cm}}\underset{n \rightarrow \infty}{\text{lim}} \mathbb{P}_{\vec\psi^*}\big(\big|\big|\widehat{\vec w} - \vec w^*\big|\big|_1 \lesssim \eta_1(n)\big) = 1,
\end{align*}
where $||\cdot||_1$ denotes the $L_1$ norm of a vector.   
\end{assumption}

\begin{assumption}[Concentration of the gradient and Hessian]\label{asmp6}
 We assume $||\vec S_{\vec\psi^*}||_\infty = O_{\mathbb{P}}(\sqrt{\log K}/n)$ and
\begin{align*}
    \Big|\Big| (1, -\vec w^{*\top})\vec S_{\vec\psi^*}\vec S_{\vec\psi^*}^\top - \mathbb{E}_{\vec\psi^*}\big[(1, -\vec w^{*\top})\vec S_{\vec\psi^*}\vec S_{\vec\psi^*}^\top\big]\Big|\Big|_\infty = O_{\mathbb{P}}(\sqrt{\log K}/n).
\end{align*}   
\end{assumption}
This assumption imposes the sub-exponential conditions for some random variables related to the gradient and Hessian matrix.

\begin{assumption}[Local smoothness conditions]\label{asmp7}
 Let $\widehat{\vec\psi}_0 = (0, \widehat{\vec\nu}_k^\top)^\top$. We assume that for both $\widecheck{\vec\psi} = \widehat{\vec\psi}_0$ and $\widecheck{\vec\psi} = \widehat{\vec\psi}$ the following holds
\begin{align*}
    (1, -\vec w^{*\top})\{\vec S_{\widecheck{\vec\psi}} - \vec S_{\vec\psi^*} - \vec S_{\vec\psi^*}\vec S_{\vec\psi^*}^\top(\widecheck{\vec\psi} - \vec\psi^*)\} &= o_{\mathbb{P}}(n^{-1/2}), \text{ and}\\
    \{(1, -\widehat{\vec w}^\top) - (1, -\vec w^{*\top})\}(\vec S_{\widecheck{\vec\psi}} - \vec S_{\vec\psi^*}) &= o_{\mathbb{P}}(n^{-1/2}).
\end{align*}   
\end{assumption}

\begin{assumption}[Central limit theorem for the efficient score function]\label{asmp8}
  We assume it holds that
\begin{align*}
    \sqrt{n}(1, -\vec w^{*\top}) \vec S_{\vec\psi^*} /\sqrt{\sigma_S^*} \sim N(0,1), \text{ where }\sigma_S^* = (1, -\widehat{\vec w}^\top)
    \Big[\underset{n \rightarrow \infty}{\text{lim}}\text{Var}(n^{1/2}\vec S_{\vec\psi^*})\Big]
    (1, -\widehat{\vec w}^\top)^\top
\end{align*}
and $\sigma_S^* \geq C$ for some constant $C > 0$.  
\end{assumption}

We follow \cite{ning2017general} to state and prove the asymptotic normality of the decorrelated score function and the one-step improved penalized G-estimator.

\begin{theorem}[\bf Asymptotic normality of the decorrelated score function] \label{thm.asymp.normal.decorr.score}
We define the score test statistic for the hypothesis $H_0: \psi_k = 0$ as $\widehat{T}_n = n^{1/2}\widehat{\ddot{S}}(0,\widehat{\vec\nu}_k,\widehat{\vec\delta})/\sqrt{\widehat{\sigma}_S}$, where $\widehat{\sigma}_S$ is a consistent estimator of $\sigma_S^*$. Then under the regularity conditions C1-C6 and the Assumptions \ref{asmp1}-\ref{asmp8}, if $\{\eta_1(n) + \eta_2(n)\}\sqrt{\log K} = o(1)$, we have
    \begin{align}
        n^{1/2}\widehat{\ddot{S}}(0,\widehat{\vec\nu}_k,\widehat{\vec\delta})\sigma_S^{*-1/2} \sim N(0,1),        
    \end{align}
 and for any $t\in \mathbb{R}$,
 \begin{align}
    \underset{n \rightarrow \infty}{\text{lim}} |\mathbb{P}_{\psi^*}(\widehat{T}_n \leq t) - \Phi(t)| = 0,
 \end{align}
 where $\Phi$ denotes the cumulative distribution function of the standard normal distribution.
\end{theorem}
\begin{proof}
    In Assumption~\ref{asmp7}, we defined that $\widehat{\vec\psi}_0 = (0, \widehat{\vec\nu}_k^\top)^\top$. Let $\vec S_{\widehat{\vec\psi}_0} = \vec S_{\vec\psi}(\vec\theta)|_{\vec\psi=\widehat{\vec\psi}_0, \vec\delta=\widehat{\vec\delta}}$. Now by the definition of $\widehat{\ddot{S}}(\widehat{\vec\psi}_0,\widehat{\vec\delta}) = \widehat{\ddot{S}}(0,\widehat{\vec\nu}_k,\widehat{\vec\delta})$, we can do the following decomposition:
    \begin{align}
        n^{1/2}|
        &\widehat{\ddot{S}}(\widehat{\vec\psi}_0,\widehat{\vec\delta}) - \ddot{S}(\vec\psi^*,\vec\delta^*)|\nonumber\\
        &=  n^{1/2}|(1, -\widehat{\vec w}^\top)\vec S_{\widehat{\vec\psi}_0} - (1, -\vec w^{*\top})\vec S_{\vec\psi^*}|\nonumber\\
        &\leq n^{1/2}|(1, -\vec w^{*\top})(\vec S_{\widehat{\vec\psi}_0} - \vec S_{\vec\psi^*})| + n^{1/2}|\{(1, -\widehat{\vec w}^\top) - (1, -\vec w^{*\top})\}\vec S_{\widehat{\vec\psi}_0}|\nonumber\\
        &=I_1+I_2. \label{eq.decorr.score.decomp}
    \end{align}
    Applying Assumption~\ref{asmp7}, we can show that
    \begin{align*}
        |I_1| &\leq n^{1/2}|(1, -\vec w^{*\top})\vec S_{\vec\psi^*}\vec S_{\vec\psi^*}^\top(\widehat{\vec\psi}_0 - \vec\psi^*)| + o_{\mathbb{P}}(1)\\
        &\leq n^{1/2}||(\widehat{\vec\psi}_0 - \vec\psi^* )||_1 ||S_{\psi_k^*}\vec S_{\vec\nu_k^*}^\top - \vec w^{*\top}\vec S_{\vec\nu_k^*}\vec S_{\vec\nu_k^*}^\top||_\infty + o_{\mathbb{P}}(1).
    \end{align*}
    By Assumptions \ref{asmp5} and \ref{asmp6}, we have $|I_1| \lesssim \eta_1(n)\sqrt{\log K} + o_{\mathbb{P}}(1) = o_{\mathbb{P}}(1)$, and Assumption~\ref{asmp7} implies that
    \begin{align*}
        |I_2| &\leq n^{1/2}|\{(1, -\widehat{\vec w}^\top) - (1, -\vec w^{*\top})\}\vec S_{\vec\psi^*}| + o_{\mathbb{P}}(1)\\
        &\leq n^{1/2}||(1, -\widehat{\vec w}^\top) - (1, -\vec w^{*\top})||_1 ||\vec S_{\vec\psi^*}||_\infty + o_{\mathbb{P}}(1).
    \end{align*}
    By Assumptions \ref{asmp5} and \ref{asmp6}, we have $|I_2| \lesssim \eta_2(n)\sqrt{\log K} + o_{\mathbb{P}}(1) = o_{\mathbb{P}}(1)$. Together with (\ref{eq.decorr.score.decomp}), the bounds for $I_1$ and $I_2$ imply $n^{1/2}|\widehat{\ddot{S}}(\widehat{\vec\psi}_0,\widehat{\vec\delta}) - \widehat{\ddot{S}}(\vec\psi^*,\vec\delta^*)| = o_{\mathbb{P}}(1)$. By Assumption~\ref{asmp8}, we have $n^{1/2}\widehat{\ddot{S}}(\vec\psi^*,\vec\delta^*)\sigma_S^{*-1/2} \sim N(0,1)$. Since $\sigma_S^* \geq C$ in Assumption~\ref{asmp8}, we have that
    \begin{align*}
        n^{1/2}|\widehat{\ddot{S}}(0, \widehat{\vec\nu}_k, \widehat{\vec\delta})\sigma_S^{*-1/2} - \ddot{S}(0, \vec\nu_k^*,\vec\delta^*)\sigma_S^{*-1/2}| = o_{\mathbb{P}}(1).
    \end{align*}
    Then applying Slutsky's theorem we complete the proof.
\end{proof}

\begin{proof}[\bf Proof of Theorem~\ref{thm.asymp.normal.OS}]
 Our goal is to show that
 \begin{align}
     n^{1/2}\big|(\widetilde{\psi}_k - \psi_k^*)I_{\psi_k|\vec\nu_k}^*/\sigma_S^{*1/2} + (1, -\vec w^{*\top}) \vec S_{\vec\psi^*}/\sigma_S^{*1/2}\big| = o_{\mathbb{P}}(1).
 \end{align}
 By the definition of $\widetilde{\psi}_k$, we have the following decomposition:
  \begin{align*}
     n^{1/2}\big|
     &(\widetilde{\psi}_k - \psi_k^*) I_{\psi_k|\vec\nu_k}^* + (1, -\vec w^{*\top}) \vec S_{\vec\psi^*}\big|\\
     &= n^{1/2}\big|(\widehat{\psi}_k - \psi_k^*) I_{\psi_k|\vec\nu_k}^* - I_{\psi_k|\vec\nu_k}^* \widehat{I}_{\psi_k|\vec\nu_k}^{-1}(1, -\widehat{\vec w}^\top)\vec S_{\widehat{\vec\psi}} + (1, -\vec w^{*\top}) \vec S_{\vec\psi^*}\big|\\
     &\leq  n^{1/2}\big|(\widehat{\psi}_k - \psi_k^*) I_{\psi_k|\vec\nu_k}^* - (1, -\vec w^{*\top})(\vec S_{\widehat{\vec\psi}} - \vec S_{\vec\psi^*})\big|\\
     &+n^{1/2}\big| \{(1, -\widehat{\vec w}^\top) - (1, -\vec w^{*\top})\}\vec S_{\widehat{\vec\psi}}\big| + n^{1/2}\big| (I_{\psi_k|\vec\nu_k}^* \widehat{I}_{\psi_k|\vec\nu_k}^{-1} - 1)(1, -\widehat{\vec w}^\top)\vec S_{\widehat{\vec\psi}}\big|\\
     &=I_1 + I_2 + I_3.
 \end{align*}
 The proof of Theorem~\ref{thm.asymp.normal.decorr.score} implies that $n^{1/2}(1, -\widehat{\vec w}^\top)\vec S_{\widehat{\vec\psi}}/\sigma_S^{*1/2} = O_{\mathbb{P}}(1).$ Thus, by the consistency of $\widehat{I}_{\psi_k|\vec\nu_k}$, we have $I_3/\sigma_S^{*1/2} = o_{\mathbb{P}}(1)$. Similar to the proof of Theorem~\ref{thm.asymp.normal.decorr.score}, we can show that
 \begin{align*}
     |I_2| \lesssim \eta_2(n)\sqrt{\log K} + o_{\mathbb{P}}(1) = o_{\mathbb{P}}(1).
 \end{align*}
 Let $S_{\psi_k^*} = \vec S_{\psi_k}(\vec\theta^*)$ and $\vec S_{\vec\nu_k^*} = \vec S_{\vec\nu_k}(\vec\theta^*)$ be the sub-vectors of $S(\vec\theta^*)$ corresponding to $\psi_k$ and $\vec\nu_k$, respectively. Next, applying the smoothness condition in Assumption~\ref{asmp7} we can show that
 \begin{align*}
     |I_1|
     &\leq n^{1/2}|(\widehat{\psi}_k - \psi_k^*) I_{\psi_k|\vec\nu_k}^* - (1, -\vec w^{*\top})\vec S_{\vec\psi^*}\vec S_{\vec\psi^*}^\top(\widehat{\vec\psi} - \vec\psi^*)| + o_{\mathbb{P}}(1)\\
     &\leq n^{1/2}|(\widehat{\psi}_k - \psi_k^*) I_{\psi_k|\vec\nu_k}^* - (\widehat{\psi}_k - \psi_k^*)( S_{\psi_k^*} S_{\psi_k^*}^\top - \vec w^{*\top}\vec S_{\vec\nu_k^*} S_{\psi_k^*}^\top)| \\
     &\hspace{0.5cm}+ n^{1/2}|(\widehat{\vec\nu}_k - \vec\nu_k^*)(S_{\psi_k^*}\vec S_{\vec\nu_k^*}^\top - \vec w^{*\top}\vec S_{\vec\nu_k^*}\vec S_{\vec\nu_k^*}^\top)| + o_{\mathbb{P}}(1)\\
     &\lesssim n^{1/2}||\widehat{\vec\psi} - \vec\psi^*||_1 ||\vec X||_\infty + o_{\mathbb{P}}(1),
 \end{align*}
 where $\vec X = [I_{\psi_k|\vec\nu_k}^* - (S_{\psi_k^*} S_{\psi_k^*}^\top - \vec w^{*\top}\vec S_{\vec\nu_k^*} S_{\psi_k^*}^\top), S_{\psi_k^*}\vec S_{\vec\nu_k^*}^\top - \vec w^{*\top}\vec S_{\vec\nu_k^*}\vec S_{\vec\nu_k^*}^\top]$ is a $K$-dimensional vector. Since by Assumption~\ref{asmp6}, $||\vec X||_\infty \lesssim \sqrt{\log K/n}$, so
 \begin{align*}
     |I_1| \lesssim \eta_1(n)\sqrt{\log K} + o_{\mathbb{P}}(1) = o_{\mathbb{P}}(1).
 \end{align*}
\end{proof}
This completes the proof.

\section{Appendix: Additional numerical results}\label{A.add.sim.res}
\setcounter{table}{0}
\setcounter{theorem}{0}
\renewcommand{\thetable}{B\arabic{table}}
\renewcommand\thetheorem{\text{B}\arabic{theorem}}

\subsection{Supplementary results for the primary simulations}
The model selection performance of the initial penalized G-estimator for all simulation scenarios considered in the manuscript are given in Table~\ref{tab.selection}.

\begin{table}[ht]
\centering
\caption{Model selection performance of the initial penalized G-estimator.}
\begin{tabular}{rrrcccc}
  \hline
 &  & corstr & FN & FP & EXACT & AFP \\ 
  \hline
$K=20$ & $n=500$ & Indep & 12.67 & 1.33 & 86.00 & 1.33 \\ 
   &  & Exch & 12.67 & 1.33 & 86.00 & 1.33 \\ 
   &  & UN & 10.67 & 0.67 & 88.67 & 0.67 \\ 
   & $n=800$ & Indep & 2.67 & 2.67 & 94.67 & 2.67 \\ 
   &  & Exch & 2.67 & 2.00 & 95.33 & 2.00 \\ 
   &  & UN & 2.67 & 2.00 & 95.33 & 2.00 \\ 
   & $n=1200$ & Indep & 0.00 & 0.67 & 99.33 & 0.67 \\ 
   &  & Exch & 0.00 & 2.00 & 98.00 & 2.00 \\ 
   &  & UN & 0.00 & 1.33 & 98.67 & 1.33 \\ 
  $K=50$ & $n=500$ & Indep & 22.00 & 0.00 & 78.00 & 0.00 \\ 
   &  & Exch & 18.00 & 0.00 & 82.00 & 0.00 \\ 
   &  & UN & 16.00 & 0.00 & 84.00 & 0.00 \\ 
   & $n=800$ & Indep & 0.67 & 0.00 & 99.33 & 0.00 \\ 
   &  & Exch & 1.33 & 0.67 & 98.00 & 0.67 \\ 
   &  & UN & 1.33 & 0.67 & 98.00 & 0.67 \\ 
   & $n=1200$ & Indep & 0.00 & 2.67 & 97.33 & 2.67 \\ 
   &  & Exch & 0.67 & 3.33 & 96.00 & 3.33 \\ 
   &  & UN & 0.67 & 4.00 & 95.33 & 4.00 \\ 
  $K=100$ & $n=500$ & Indep & 32.00 & 0.00 & 68.00 & 0.00 \\ 
   &  & Exch & 29.33 & 0.00 & 70.67 & 0.00 \\ 
   &  & UN & 30.67 & 0.00 & 69.33 & 0.00 \\ 
   & $n=800$ & Indep & 3.33 & 0.00 & 96.67 & 0.00 \\ 
   &  & Exch & 3.33 & 0.00 & 96.67 & 0.00 \\ 
   &  & UN & 3.33 & 0.00 & 96.67 & 0.00 \\ 
   & $n=1200$ & Indep & 0.00 & 0.00 & 100.00 & 0.00 \\ 
   &  & Exch & 0.00 & 0.67 & 99.33 & 0.67 \\ 
   &  & UN & 0.00 & 0.67 & 99.33 & 0.67 \\ 
   \hline
\end{tabular}\\ \label{tab.selection}
\footnotesize
FN: \% of false negatives, FP: \% of false positives, EXACT: \% of exact selections,\\
AFP: average false positives, Indep: independent, Exch: exchangeable, UN: unstructured
\end{table}

\subsection{Investigation on the G-null paradox}\label{A.gNULL}

The G-null paradox demonstrates that even when the global null hypothesis of no treatment effect is true, G-estimation can produce biased, non-zero effect estimates if the treatment-free outcome model is misspecified. This is true in case of doubly-robust estimation if the treatment-model is incorrect and the estimated propensity scores are biased. In this section, we investigate the performance of the penalized G-estimator and the proposed inferential methods considering a simulation setting which is designed to trigger the G-null paradox.

Data were generated assuming an exchangeable correlation structure with $\alpha = 0.8$, error variance $\sigma^2_\epsilon = 1$, and autocorrelation coefficient $\rho = 0.25$. The sample size was set to $n = 500$, with $J = 6$ visits for all subjects. All other data-generating steps followed those described in the manuscript, except for the specifications of the treatment-free model and the outcome model coefficients. The true treatment-free model, designed to induce the G-null paradox, is as follows:
\begin{align*}
    \mu_j(\vec h_j; \vec\delta)&=\delta_0+\delta_1 l^{(1)}+\delta_2 l^{(2)}
    +\sum_{m=3}^6 \delta_m l_j^{(m)}
    +\sum_{m=1}^{20} \delta_{6+m}\,x_j^{(m)}
    +\sum_{m=21}^{K-6} \delta_{6+m}\,x_j^{(m)}
    +\delta_{K+1} \exp(l^{(6)})
\end{align*}
The true blip function is $\gamma^*_j(a_j,\vec h_j; \vec\psi)=(\psi_0+\psi_1 l^{(1)}+\psi_2 l^{(2)}+\sum_{m=3}^6\psi_m l_j^{(m)}+\sum_{m=1}^{20} \psi_{6+m}\,x_j^{(m)}+\sum_{m=21}^{K-6} \psi_{6+m}\,x_j^{(m)})a_j$ with common parameters at each time point. The true outcome model coefficients are $\vec\delta=(\delta_0,\ldots,\delta_{K+1})^\top$ and $\vec\psi=(\psi_0,\ldots,\psi_{K})^\top$, where
\begin{align*}
   \vec\delta&=(1,1, 1.2, 1.2, -0.9, 0.8, -1, 1, \ldots, 1, 0, \ldots,0,1.1)^\top\\
   \vec\psi&=(1, 2.5, -2.5, -2.8, 2.6, 2.8, 0, 0, \ldots,0, 0, \ldots,0)^\top.
\end{align*}
Note that $l^{(6)}$ is a true confounder but not an effect modifier, since $\psi_6 = 0$ in the blip function. This setup was designed so that, when both the treatment-free model and the treatment model are misspecified, the penalized G-estimation yields a biased (non-zero) estimate for $\psi_6$. This allows us to assess whether the proposed inferential methods can avoid the G-null paradox—that is, whether they yield valid inferences when at least one of the treatment or treatment-free models is correctly specified.

We evaluated the estimation performance under four scenarios: Scenario 1 (TcTFw) – the treatment model is correctly specified, while the treatment-free model is misspecified; Scenario 2 (TwTFc) – the treatment model is misspecified, while the treatment-free model is correctly specified; Scenario 3 (TcTFc) – both models are correctly specified; and Scenario 4 (TwTFw) – both models are misspecified. For each generated dataset, the proposed estimation was performed under each of the four scenarios and three working correlation structures. In Scenario 1, the propensity score was estimated using the true exposure model, but the $\exp(L^{(6)})$ predictor was excluded from the treatment-free model. In Scenario 2, $\exp(L^{(6)})$ was included in the treatment-free model, but confounder $L^{(6)}$ was excluded from the exposure model used for propensity score estimation. In Scenario 3, both the true exposure model and the correctly specified treatment-free model (including $\exp(L^{(6)})$) were used. Finally, in Scenario 4, confounder $L^{(6)}$ was excluded from the exposure model, and $\exp(L^{(6)})$ was excluded from the treatment-free component. In penalized G-estimation, the hyperparameter in the derivative function of the SCAD penalty was set to 2. We report, in Table~\ref{tab_gNull_selection}, the percentage of times each candidate effect modifier was selected, along with the false negative, false positive, and exact selection rates. Tables \ref{tab_gNull_fcr_coef} and \ref{tab_gNull_al_coef} summarize the performance of the inferential methods with respect to coefficientwise false coverage proportions and average confidence interval lengths, respectively. Since UPoSI method generally provides overly wide confidence intervals, it was excluded from the comparison.

\begin{table}[ht]
\centering
\caption{Effect modifier selection performance of the penalized G-estimator under different correlation structures (Ind: Independent, Exch: Exchangeable, UN: Unstructured) and estimation scenarios, based on 100 simulations with $K=100$, $n=500$, $\rho = 0.25$, $\sigma_{\epsilon}^2 = 1$, and $\alpha = 0.8$ (true correlation structure: exchangeable).}
\resizebox{1\textwidth}{!}{
\begin{tabular}{lrrrlrrrlrrrlrrr}
  \toprule
  &\multicolumn{3}{c}{\bf Scenario 1 (TcTFw)}&&\multicolumn{3}{c}{\bf Scenario 2 (TwTFc)}&&\multicolumn{3}{c}{\bf Scenario 3 (TcTFc)}&&\multicolumn{3}{c}{\bf Scenario 4 (TwTFw)}\\
  \cline{2-4}\cline{6-8}\cline{10-12}\cline{14-16}
 &\bf Indep &\bf Exch &\bf UN & &\bf Indep &\bf Exch &\bf UN & &\bf Indep &\bf Exch &\bf UN & &\bf Indep &\bf Exch &\bf UN\\ 
  \hline
$A\times L^{(1)}$ & 100 & 100 & 100 &  & 100 & 100 & 100 &  & 100 & 100 & 100 &  & 100 & 100 & 100 \\ 
$A\times L^{(2)}$  & 100 & 100 & 99 &  & 100 & 100 & 100 &  & 100 & 100 & 100 &  & 100 & 100 & 98 \\ 
$A\times L^{(3)}$  & 100 & 100 & 98 &  & 100 & 100 & 100 &  & 100 & 100 & 100 &  & 100 & 100 & 99 \\ 
$A\times L^{(4)}$  & 100 & 100 & 99 &  & 100 & 100 & 100 &  & 100 & 100 & 100 &  & 100 & 100 & 100 \\ 
$A\times L^{(5)}$  & 100 & 100 & 98 &  & 100 & 100 & 100 &  & 100 & 100 & 100 &  & 100 & 100 & 98 \\ 
$A\times L^{(6)}$  & 4 & 4 & 4 &  & 0 & 0 & 0 &  & 0 & 0 & 0 &  & 100 & 100 & 100 \\ 
$A\times \exp(L^{(6)})$  & - & - & - &  & 0 & 0 & 0 &  & 0 & 0 & 0 &  & - & - & - \\ 
$A\times X^{(1)}$ & 0 & 0 & 0 &  & 0 & 0 & 0 &  & 0 & 0 & 0 &  & 0 & 0 & 0 \\ 
$A\times X^{(2)}$ & 0 & 0 & 0 &  & 0 & 0 & 0 &  & 0 & 0 & 0 &  & 0 & 0 & 0 \\ 
$A\times X^{(3)}$ & 0 & 0 & 0 &  & 0 & 0 & 0 &  & 0 & 0 & 0 &  & 0 & 0 & 0 \\ 
$A\times X^{(4)}$ & 0 & 0 & 0 &  & 0 & 0 & 0 &  & 0 & 0 & 0 &  & 0 & 0 & 0 \\ 
$A\times X^{(5)}$ & 0 & 0 & 0 &  & 0 & 0 & 0 &  & 0 & 0 & 0 &  & 0 & 0 & 0 \\ 
$A\times X^{(6)}$ & 0 & 0 & 0 &  & 0 & 0 & 0 &  & 0 & 0 & 0 &  & 0 & 0 & 0 \\ 
$A\times X^{(7)}$ & 0 & 0 & 0 &  & 0 & 0 & 0 &  & 0 & 0 & 0 &  & 0 & 0 & 0 \\ 
$A\times X^{(8)}$ & 0 & 0 & 0 &  & 0 & 0 & 0 &  & 0 & 0 & 0 &  & 0 & 0 & 0 \\ 
$A\times X^{(9)}$ & 0 & 0 & 0 &  & 0 & 0 & 0 &  & 0 & 0 & 0 &  & 0 & 0 & 0 \\ 
$A\times X^{(10)}$ & - & - & - &  & 0 & 0 & 0 &  & 0 & 0 & 0 &  & - & - & -  \\ 
\vdots&\vdots&\vdots&\vdots&&\vdots&\vdots&\vdots&&\vdots&\vdots&\vdots&&\vdots&\vdots&\vdots\\
$A\times X^{(94)}$ & 0 & 0 & 0 &  & 0 & 0 & 0 &  & 0 & 0 & 0 &  & 0 & 0 & 0 \\ 
$A\times X^{(95)}$ & 0 & 0 & 0 &  & 0 & 0 & 0 &  & 0 & 0 & 0 &  & 0 & 0 & 0 \\ 
  FN & 0 & 0 & 2 &  & 0 & 0 & 0 &  & 0 & 0 & 0 &  & 0 & 0 & 2 \\ 
  FP & 4 & 4 & 4 &  & 0 & 0 & 0 &  & 2 & 0 & 0 &  & 100 & 100 & 100 \\ 
  EXACT & 96 & 96 & 94 &  & 100 & 100 & 100 &  & 98 & 100 & 100 &  & 0 & 0 & 0 \\ 
   \bottomrule
\end{tabular}
}\label{tab_gNull_selection}
\vspace{0.8cm}
\footnotesize
FN: \% of false negatives, FP: \% of false positives, EXACT: \% of exact selections
\end{table}

Table~\ref{tab_gNull_selection} indicates that when both models were misspecified (Scenario 4), the penalized G-estimation incorrectly selected $L^{(6)}$ as an effect modifier in 100\% of the simulations, leading to false positives. In contrast, when at least one of the models was correctly specified (Scenarios 1–3), exact selection occurred in over 90\% of the simulations. Regarding the inferential methods, the naive Wald-type approach (based on the sandwich variance estimator) produced overly narrow confidence intervals (see Table~\ref{tab_gNull_al_coef}) and exhibited false coverage proportions exceeding the nominal 0.05 level (see Table~\ref{tab_gNull_fcr_coef}) in most of the cases. The decorrelated score approach based on sparse weights obtained via the Dantzig selector maintained control of false coverage rates across all correlation structures in Scenarios 1-3. Under Scenario 1, the variable $L^{(6)}$ was incorrectly selected as an effect modifier in 4\% of the simulations; in those cases, Dantzig-based inference attained adequate coverage for the truly null coefficient $\psi_6$, whereas the naive confidence intervals rarely included the null value. When both models were misspecified (Scenario~4), however, the decorrelated score approach exhibited false coverage proportions exceeding 5\% for main effect $\psi_0$ and the true null coefficient $\psi_6$. Hence, we expect that accurate specification of at least one model will reduce the risk of structural model misspecification and will help approximate the true data-generating processes more closely, thereby mitigating the impact of the G-null paradox in practice.

\newcolumntype{L}[1]{>{\hspace{#1}}l}

\begin{table}[ht]
\centering
\caption{Coefficient-wise false coverage results for the inferential methods under different correlation structures (Ind: Independent, Exch: Exchangeable, UN: Unstructured) and estimation scenarios, based on 100 simulations with $K=100$, $n=500$, $\rho = 0.25$, $\sigma_{\epsilon}^2 = 1$, and $\alpha = 0.8$ (true correlation structure: exchangeable).}
\resizebox{1\textwidth}{!}{
\begin{tabular}{L{1em}rrrlrrrlrrrlrrr}
  \toprule
  &\multicolumn{3}{c}{\bf Scenario 1 (TcTFw)}&&\multicolumn{3}{c}{\bf Scenario 2 (TwTFc)}&&\multicolumn{3}{c}{\bf Scenario 3 (TcTFc)}&&\multicolumn{3}{c}{\bf Scenario 4 (TwTFw)}\\
  \cline{2-4}\cline{6-8}\cline{10-12}\cline{14-16}
 &\bf Indep &\bf Exch &\bf UN & &\bf Indep &\bf Exch &\bf UN & &\bf Indep &\bf Exch &\bf UN & &\bf Indep &\bf Exch &\bf UN\\
\hline
\hspace{-1em}$\psi_0 = 1$\\
Naive & 0.07 & 0.05 & 0.07 &  & 0.08 & 0.10 & 0.13 &  & 0.09 & 0.07 & 0.07 &  & 0.64 & 0.85 & 0.84 \\ 
  OS.Full & 0.00 & 0.01 & 0.04 &  & 0.09 & 0.07 & 0.09 &  & 0.10 & 0.08 & 0.08 &  & 0.22 & 0.20 & 0.22 \\ 
  OS.LASSO & 0.01 & 0.00 & 0.03 &  & 0.01 & 0.02 & 0.03 &  & 0.04 & 0.02 & 0.02 &  & 0.09 & 0.07 & 0.09 \\ 
  OS.Dantzig & 0.00 & 0.00 & 0.02 &  & 0.00 & 0.00 & 0.00 &  & 0.00 & 0.00 & 0.00 &  & 0.08 & 0.04 & 0.07\vspace{0.2cm}\\
  \hspace{-1em}$\psi_1 = 2.5$\\
  Naive & 0.12 & 0.12 & 0.09 &  & 0.13 & 0.09 & 0.07 &  & 0.12 & 0.03 & 0.05 &  & 0.11 & 0.15 & 0.13 \\ 
  OS.Full & 0.06 & 0.08 & 0.08 &  & 0.18 & 0.05 & 0.08 &  & 0.10 & 0.06 & 0.06 &  & 0.11 & 0.13 & 0.12 \\ 
  OS.LASSO & 0.01 & 0.02 & 0.01 &  & 0.06 & 0.03 & 0.01 &  & 0.02 & 0.03 & 0.03 &  & 0.07 & 0.04 & 0.04 \\ 
  OS.Dantzig & 0.00 & 0.00 & 0.00 &  & 0.03 & 0.00 & 0.00 &  & 0.00 & 0.00 & 0.00 &  & 0.01 & 0.01 & 0.01\vspace{0.2cm}\\
  \hspace{-1em}$\psi_2 = -2.5$\\
  Naive & 0.04 & 0.05 & 0.07 &  & 0.06 & 0.05 & 0.03 &  & 0.04 & 0.04 & 0.04 &  & 0.08 & 0.10 & 0.10 \\ 
  OS.Full & 0.06 & 0.04 & 0.08 &  & 0.06 & 0.04 & 0.08 &  & 0.02 & 0.04 & 0.07 &  & 0.09 & 0.13 & 0.11 \\ 
  OS.LASSO & 0.02 & 0.01 & 0.03 &  & 0.01 & 0.01 & 0.02 &  & 0.01 & 0.01 & 0.01 &  & 0.07 & 0.05 & 0.05 \\ 
  OS.Dantzig & 0.00 & 0.00 & 0.01 &  & 0.00 & 0.00 & 0.00 &  & 0.00 & 0.00 & 0.00 &  & 0.00 & 0.00 & 0.00\vspace{0.2cm}\\
  \hspace{-1em}$\psi_3 = -2.8$\\
  Naive & 0.15 & 0.10 & 0.10 &  & 0.11 & 0.12 & 0.11 &  & 0.11 & 0.11 & 0.13 &  & 0.14 & 0.14 & 0.15 \\ 
  OS.Full & 0.07 & 0.09 & 0.11 &  & 0.07 & 0.07 & 0.09 &  & 0.12 & 0.08 & 0.10 &  & 0.24 & 0.20 & 0.19 \\ 
  OS.LASSO & 0.00 & 0.02 & 0.06 &  & 0.04 & 0.01 & 0.05 &  & 0.01 & 0.01 & 0.03 &  & 0.08 & 0.14 & 0.14 \\ 
  OS.Dantzig & 0.00 & 0.00 & 0.00 &  & 0.00 & 0.01 & 0.01 &  & 0.00 & 0.00 & 0.00 &  & 0.00 & 0.00 & 0.01\vspace{0.2cm}\\
  \hspace{-1em}$\psi_4 = 2.6$\\
  Naive & 0.09 & 0.09 & 0.07 &  & 0.09 & 0.14 & 0.11 &  & 0.09 & 0.09 & 0.05 &  & 0.11 & 0.12 & 0.12 \\ 
  OS.Full & 0.09 & 0.09 & 0.11 &  & 0.11 & 0.07 & 0.09 &  & 0.11 & 0.07 & 0.09 &  & 0.09 & 0.09 & 0.14 \\ 
  OS.LASSO & 0.02 & 0.04 & 0.04 &  & 0.05 & 0.00 & 0.02 &  & 0.01 & 0.03 & 0.04 &  & 0.06 & 0.04 & 0.10 \\ 
  OS.Dantzig & 0.00 & 0.00 & 0.01 &  & 0.01 & 0.00 & 0.00 &  & 0.01 & 0.00 & 0.00 &  & 0.00 & 0.00 & 0.01\vspace{0.2cm}\\
  \hspace{-1em}$\psi_5 = 2.8$\\
  Naive & 0.06 & 0.06 & 0.05 &  & 0.06 & 0.09 & 0.09 &  & 0.05 & 0.08 & 0.07 &  & 0.06 & 0.06 & 0.03 \\ 
  OS.Full & 0.13 & 0.20 & 0.15 &  & 0.11 & 0.06 & 0.06 &  & 0.07 & 0.08 & 0.09 &  & 0.12 & 0.11 & 0.15 \\ 
  OS.LASSO & 0.09 & 0.08 & 0.09 &  & 0.03 & 0.01 & 0.02 &  & 0.01 & 0.05 & 0.02 &  & 0.07 & 0.10 & 0.07 \\ 
  OS.Dantzig & 0.00 & 0.00 & 0.00 &  & 0.00 & 0.00 & 0.00 &  & 0.00 & 0.00 & 0.00 &  & 0.00 & 0.00 & 0.00\vspace{0.2cm}\\
  \hspace{-1em}$\psi_6 = 0$\\
  Naive & 0.75 & 1.00 & 0.75 &  & - & - & - &  & - & - & - &  & 1.00 & 1.00 & 1.00 \\ 
  OS.Full & 0.00 & 0.25 & 0.50 &  & - & - & - &  & - & - & - &  & 1.00 & 1.00 & 1.00 \\ 
  OS.LASSO & 0.00 & 0.00 & 0.50 &  & - & - & - &  & - & - & - &  & 1.00 & 1.00 & 1.00 \\ 
  OS.Dantzig & 0.00 & 0.00 & 0.00 &  & - & - & - &  & - & - & - &  & 1.00 & 1.00 & 1.00\\ 
\bottomrule
\end{tabular}
}\label{tab_gNull_fcr_coef}
\end{table}

\begin{table}[ht]
\centering
\caption{Coefficient-wise average CI length for the inferential methods under different correlation structures (Ind: Independent, Exch: Exchangeable, UN: Unstructured) and estimation scenarios, based on 100 simulations with $K=100$, $n=500$, $\rho = 0.25$, $\sigma_{\epsilon}^2 = 1$, and $\alpha = 0.8$ (true correlation structure: exchangeable).}
\resizebox{1\textwidth}{!}{
\begin{tabular}{L{1em}rrrlrrrlrrrlrrr}
  \toprule
  &\multicolumn{3}{c}{\bf Scenario 1 (TcTFw)}&&\multicolumn{3}{c}{\bf Scenario 2 (TwTFc)}&&\multicolumn{3}{c}{\bf Scenario 3 (TcTFc)}&&\multicolumn{3}{c}{\bf Scenario 4 (TwTFw)}\\
  \cline{2-4}\cline{6-8}\cline{10-12}\cline{14-16}
 &\bf Indep &\bf Exch &\bf UN & &\bf Indep &\bf Exch &\bf UN & &\bf Indep &\bf Exch &\bf UN & &\bf Indep &\bf Exch &\bf UN\\
\hline
\hspace{-1em}$\psi_0 = 1$\\
Naive & 0.41 & 0.39 & 0.41 &  & 0.19 & 0.09 & 0.09 &  & 0.20 & 0.10 & 0.10 &  & 0.64 & 0.45 & 0.46 \\ 
  OS.Full & 1.04 & 1.04 & 0.99 &  & 0.33 & 0.16 & 0.15 &  & 0.44 & 0.21 & 0.21 &  & 0.81 & 0.83 & 0.79 \\ 
  OS.LASSO & 1.39 & 1.36 & 1.31 &  & 0.44 & 0.19 & 0.19 &  & 0.63 & 0.26 & 0.26 &  & 1.05 & 1.05 & 1.02 \\ 
  OS.Dantzig & 1.52 & 1.55 & 1.49 &  & 0.61 & 0.26 & 0.25 &  & 0.87 & 0.40 & 0.39 &  & 1.25 & 1.26 & 1.21\vspace{0.2cm}\\
  \hspace{-1em}$\psi_1 = 2.5$\\
 Naive & 0.49 & 0.47 & 0.49 &  & 0.21 & 0.10 & 0.10 &  & 0.22 & 0.11 & 0.11 &  & 0.50 & 0.48 & 0.50 \\ 
  OS.Full & 1.18 & 1.19 & 1.13 &  & 0.27 & 0.14 & 0.13 &  & 0.33 & 0.16 & 0.15 &  & 0.88 & 0.90 & 0.87 \\ 
  OS.LASSO & 1.61 & 1.59 & 1.62 &  & 0.34 & 0.17 & 0.16 &  & 0.44 & 0.20 & 0.19 &  & 1.19 & 1.24 & 1.20 \\ 
  OS.Dantzig & 1.84 & 1.90 & 1.84 &  & 0.39 & 0.19 & 0.19 &  & 0.50 & 0.23 & 0.22 &  & 1.45 & 1.51 & 1.45\vspace{0.2cm}\\
  \hspace{-1em}$\psi_2 = -2.5$\\
  Naive & 0.51 & 0.49 & 0.49 &  & 0.21 & 0.10 & 0.10 &  & 0.22 & 0.11 & 0.11 &  & 0.54 & 0.51 & 0.50 \\ 
  OS.Full & 1.22 & 1.22 & 1.17 &  & 0.27 & 0.14 & 0.14 &  & 0.34 & 0.16 & 0.16 &  & 0.89 & 0.92 & 0.90 \\ 
  OS.LASSO & 1.67 & 1.74 & 1.69 &  & 0.35 & 0.17 & 0.17 &  & 0.45 & 0.20 & 0.20 &  & 1.18 & 1.27 & 1.22 \\ 
  OS.Dantzig & 1.98 & 2.04 & 1.96 &  & 0.40 & 0.20 & 0.20 &  & 0.52 & 0.24 & 0.23 &  & 1.51 & 1.60 & 1.52
  \vspace{0.2cm}\\
  \hspace{-1em}$\psi_3 = -2.8$\\
 Naive & 0.53 & 0.50 & 0.50 &  & 0.20 & 0.10 & 0.10 &  & 0.22 & 0.11 & 0.11 &  & 0.55 & 0.53 & 0.54 \\ 
  OS.Full & 1.22 & 1.23 & 1.18 &  & 0.28 & 0.14 & 0.14 &  & 0.34 & 0.17 & 0.16 &  & 0.91 & 0.93 & 0.89 \\ 
  OS.LASSO & 1.74 & 1.69 & 1.57 &  & 0.35 & 0.17 & 0.17 &  & 0.45 & 0.20 & 0.19 &  & 1.20 & 1.27 & 1.23 \\ 
  OS.Dantzig & 2.00 & 2.06 & 2.02 &  & 0.40 & 0.20 & 0.20 &  & 0.52 & 0.24 & 0.24 &  & 1.57 & 1.69 & 1.65\vspace{0.2cm}\\
  \hspace{-1em}$\psi_4 = 2.6$\\
  Naive & 0.47 & 0.44 & 0.46 &  & 0.19 & 0.10 & 0.10 &  & 0.21 & 0.11 & 0.11 &  & 0.49 & 0.47 & 0.48 \\ 
  OS.Full & 1.09 & 1.10 & 1.05 &  & 0.26 & 0.13 & 0.12 &  & 0.31 & 0.15 & 0.14 &  & 0.84 & 0.85 & 0.81 \\ 
  OS.LASSO & 1.46 & 1.49 & 1.46 &  & 0.32 & 0.15 & 0.14 &  & 0.40 & 0.17 & 0.17 &  & 1.09 & 1.15 & 1.12 \\ 
  OS.Dantzig & 1.67 & 1.69 & 1.68 &  & 0.38 & 0.18 & 0.17 &  & 0.46 & 0.21 & 0.20 &  & 1.38 & 1.47 & 1.38\vspace{0.2cm}\\
  \hspace{-1em}$\psi_5 = 2.8$\\
  Naive & 0.53 & 0.51 & 0.51 &  & 0.18 & 0.09 & 0.09 &  & 0.20 & 0.10 & 0.10 &  & 0.52 & 0.50 & 0.50 \\ 
  OS.Full & 1.16 & 1.17 & 1.12 &  & 0.26 & 0.13 & 0.13 &  & 0.33 & 0.16 & 0.15 &  & 0.84 & 0.87 & 0.83 \\ 
  OS.LASSO & 1.62 & 1.62 & 1.59 &  & 0.33 & 0.16 & 0.15 &  & 0.44 & 0.19 & 0.19 &  & 1.08 & 1.17 & 1.17 \\ 
  OS.Dantzig & 2.15 & 2.32 & 2.17 &  & 0.38 & 0.19 & 0.19 &  & 0.54 & 0.24 & 0.24 &  & 1.65 & 1.80 & 1.75\vspace{0.2cm}\\
  \hspace{-1em}$\psi_6 = 0$\\
   Naive & 0.71 & 0.70 & 0.72 &  & - & - & - &  & - & - & - &  & 1.07 & 0.99 & 0.98 \\ 
  OS.Full & 1.28 & 1.29 & 1.23 &  & - & - & - &  & - & - & - &  & 0.56 & 0.59 & 0.57 \\ 
  OS.LASSO & 1.53 & 1.73 & 1.66 &  & - & - & - &  & - & - & - &  & 0.72 & 0.79 & 0.76 \\ 
  OS.Dantzig & 2.69 & 3.17 & 2.85 &  & - & - & - &  & - & - & - &  &1.43 & 1.59 & 1.47\\
\bottomrule
\end{tabular}
}\label{tab_gNull_al_coef}
\end{table}
\end{document}